\documentclass{article}

\usepackage{arxiv}

\RequirePackage[colorlinks,citecolor=blue,linkcolor=blue,urlcolor=blue,pagebackref]{hyperref}
\usepackage[utf8]{inputenc} 
\usepackage[T1]{fontenc}    
\usepackage{hyperref}       
\usepackage{url}            
\usepackage{booktabs}       
\usepackage{amsfonts}       
\usepackage{nicefrac}       
\usepackage{microtype}      
\usepackage{xcolor}         
\usepackage{amsmath,amssymb}
\usepackage{natbib} 
\usepackage{comment}         
\newcommand{\annot}[2]{\underbrace{#1}_{\text{#2}}}

\usepackage{color}

\usepackage{amsthm}

\usepackage[ruled,vlined]{algorithm2e}
\usepackage{bm}

\usepackage{bbm}

\newcommand{\N}{\mathbb{N}}
\newcommand{\R}{\mathbb{R}}

\newcommand{\argmin}[1]{\underset{#1}{\text{argmin}}}

\newcommand{\eps}{\epsilon}

\newcommand{\Prob}{\mathbb{P}}

\usepackage{amsthm}
\usepackage[capitalise]{cleveref}
\theoremstyle{plain}

\newtheorem{theorem}{Theorem}[section]
\newtheorem{lemma}[theorem]{Lemma}

\newtheorem{corollary}[theorem]{Corollary}
\newtheorem{proposition}[theorem]{Proposition}

\newtheorem{definition}[theorem]{Definition}
\newtheorem{assumption}[theorem]{Assumption}

\newtheorem*{remark}{Remark}

\newcommand*{\email}[1]{\texttt{#1}}

\newcommand\Sc {\mathcal{S}}
\newcommand\Nc {\mathcal{N}}
\newcommand\bbH{\mathbb{H}}

\newcommand\E {\mathbb{E}}
\newcommand\Expect{\E}
\newcommand\Pbb {\mathbb{P}}
\renewcommand\Re{\R}
\newcommand\bs {\boldsymbol}

\newcommand\rank {\ensuremath{\operatorname{\mathrm{rank}}}}

\newcommand\tr {\ensuremath{\operatorname{\mathrm{tr}}}}

\usepackage{graphicx}
\newcommand{\indep}{\rotatebox[origin=c]{90}{$\models$}}

\title{Benign-Overfitting in Conditional Average Treatment Effect Prediction with Linear Regression}

\newcommand*{\affaddr}[1]{#1} 
\newcommand*{\affmark}[1][*]{\textsuperscript{#1}}
\newcommand*{\equalcontribution}[1][*]{\textsuperscript{*}}

\author{%
Masahiro Kato\affmark[1,2]\thanks{\email{masahiro\_kato@cyberagent.co.jp}.}\ \ \ \ \ \ Masaaki Imaizumi\affmark[1]\\
\affaddr{\affmark[1]The University of Tokyo}\\
\affaddr{\affmark[2]CyberAgent, Inc.}
}

\begin{document}

\maketitle

\begin{abstract}
We study the benign overfitting theory in the prediction of the conditional average treatment effect (CATE), with linear regression models. As the development of machine learning for causal inference, a wide range of large-scale models for causality are gaining attention. One problem is that suspicions have been raised that the large-scale models are prone to overfitting to observations with sample selection, hence the large models may not be suitable for causal prediction. In this study, to resolve the suspicious, we investigate on the validity of causal inference methods for overparameterized models, by applying the recent theory of benign overfitting \citep{Bartlett2020}. Specifically, we consider samples whose distribution switches depending on an assignment rule, and study the prediction of CATE with linear models whose dimension diverges to infinity. We focus on two methods: the T-learner, which based on a difference between separately constructed estimators with each treatment group, and the inverse probability weight (IPW)-learner, which solves another regression problem approximated by a propensity score. In both methods, the estimator consists of interpolators that fit the samples perfectly. As a result, we show that the T-learner fails to achieve the consistency except the random assignment, while the IPW-learner converges the risk to zero if the propensity score is known. This difference stems from that the T-learner is unable to preserve eigenspaces of the covariances, which is necessary for benign overfitting in the overparameterized setting. Our result provides new insights into the usage of causal inference methods in the overparameterizated setting, in particular, doubly robust estimators.

\end{abstract}

\section{Introduction}
The problem of predicting the causal effects of treatment from observations is a central task in various fields, such as economics \citep{Wager2018}, medicine \citep{Assmann2000,Foster2011}, and online advertisement \citep{Bottou13}. The exact treatment effect is a counterfactual value, and it is usually intractable to know it directly. Therefore, we are often interested in the average treatment effect (ATE), which is defined as the difference between the expected potential outcomes of the two treatments \citep{Neyman,Rubin1974,imbens_rubin_2015}. Furthermore, in the past few years, the growth of large observational data has encouraged the development of predicting the individualized treatment effect, also called the conditional ATE (CATE) \citep{hahn1998role,Heckman1997,Abrevaya2015}, to allow individuals to have different ATEs; that is, the treatment effect can be heterogeneous among individuals. 


CATE prediction is increasing importance in statistics and machine learning \citep{Qian2011,Zhao2012,Imai2013,Zhou2017}. A naive way to predict the CATE is to estimate the ATE on the subgroups \citep{Assmann2000,Cook2004}; that is, computing the ATE on each group separated based on the covariates.
When the covariates are continuous, it is common to assume some statistical model for CATE, and various methods have been proposed to learn the model \citep{Weisberg2015}. In particular, in recent years, there has been a lot of interest in how to train models given high-dimensional data and models
\citep{Belloni2011,Belloni2014,Candes2007,Moon2007,Sun2012,Song2015}. 
Especially, several machine learning studies estimate and predict CATE using large-scale flexible models such as neural networks  \citep{Johansson2016,Shalit2017,Yao2018,Atan2018,Farrell2020}.

Despite the developments, the validity of CATE prediction using large-scale models is still an ongoing issue.
When the data are not perfectly observed, as in the CATE problem, the flexible models is more likely to overfit the observations.
\citet{Kunzel2019} points out the possibility of overfitting of a method for CATE prediction. 
For ordinary regression problems without no treatment effects, new various theories for overparameterized models have emerged such as the \textit{benign overfitting} \citep{Bartlett2020}; that is, a prediction error can be sufficiently small, even though the predictor overfits training data by numerous parameters greater than the sample size.
However, it is not clear whether the recent overparameterization theory are applicable to the problems of CATE prediction.

In this study, we investigate CATE prediction with linear regression models whose number of parameter is larger than the number of observations by following \citet{Bartlett2020}. 
Specifically, we consider the excess risk of two standard prediction methods: the T-learner \citep{Kunzel2019} and inverse probability weight (IPW)-learner with linear regression models. 
The T-learner  separately constructs the interpolating estimator of each treatment effect. The IPW-learner  first approximates a response variable of the CATE problem by using treatment assignment probability, called the propensity score, then develops prediction for the approximated response. 
For both methods, we derive the upper and lower bounds on the excess risk of CATE prediction and then investigate the conditions under which the upper bound goes to zero. 

As a result, we find that the design of the treatment assignment plays an important role in benign overfitting. 
For the T-learner case, when the treatment assignment does not depend on the covariates, which is standard in \textit{randomized controlled trials} (RCTs), the prediction risk goes to zero under the same conditions as in \citet{Bartlett2020}. 
In contrast, when the treatment assignment depends on the covariates owing to the selection bias, the convergence of the risk is not guaranteed.
This result is consistent with the previous works \citep{Kunzel2019,Nie2020}, which claims the danger of overfitting of the T-learner. 
For the IPW-learner case, on the other hand, the prediction risk converges to zero, regardless of the assignment rule.
These results give implications for CATE prediction of overparameterization, and also provide insights into the use of other methods. 
For example, when using the two-step algorithms, such as doubly robust method \citep{ChernozhukovVictor2018Dmlf,Kennedy2020} and R-learner \citep{Nie2020}, our result implies the importance of correctly estimating IPW in the first stage nuisance parameter estimation instead of using the T-learner.

\paragraph{Related work.} There is a rich literature on the overparameterized setting.
The most closest theory to our work is the benign overfitting in linear regression \citet{Bartlett2020}, which reveals a sufficient condition under which the prediction risk converges to zero with the overparameterized linear model.
Subsequent to \citet{Bartlett2020}, the framework is extended to ridge regression
\citep{Tsigler2020}, 
multiclass classification \citep{wang2021benign}, and a max-margin classifier \citep{cao2021risk}.
\cite{koehler2021uniform} reveals a connection between benign overfitting and the notion of uniform convergence.
For other studies studies on overparameterization, numerous works study the precious asymptotics of overparameterized models by using the random matrix theory. 
\cite{Muthukumar2019,hastie2019surprises,dobriban2018high} consider linear regression or shallow neural networks, 
\citet{dobriban2018high,Denny2020} study a ridge regression problem,
and \cite{Chatterji2021} studies a binary classification problem, \citet{BelkinHsu20182,BelkinRakhlin2019,Liang2020} study interpolating kernel methods.


CATE prediction has also been proposed using kernel-based methods \citep{Fan2008}, Gaussian processes \citep{Alaa2017,Alaa2018}, generative adversarial nets \citep{yoon2018ganite}, boosting, tree-based methods \citep{Zeileis2008,Su2009,ImaiSt2011,Kang2012,Lipkovich2011,loh2012,Wager2018,Athey2019,Chatla2020}, nearest neighbor matching, series estimation, and Bayesiaan additive regression trees \citep{Hill2011}. \citet{Gunter2011,ImaiSt2011,Imai2013} formulate the CATE estimation problem as a variable selection problem. \cite{Cai2017,Cai2021} study confidence intervals for high-dimensional cases. As a unifying framework, 
\citet{Kunzel2019} introduces meta-learners, such as the T-learner and X-learner. Other various methods have also been proposed \citep{Li2017,Kallus2017,Powers2017,Subbaswamy2018,Zhao2019,Hahn2020,Nie2020,Kennedy2020}.


\paragraph{Notation.}
We define a (potentially infinite-dimensional) Hilbert space $\bbH$ with a norm $\|\cdot\|$.
For two vectors $z,z' \in \mathbbm{H}$, $z^\top z'$ denotes an inner product between $z$ and $z'$, and $z z'^{\top}$ denotes a tensor product. 
For an operator $\Sigma: \mathbb{H} \to \mathbb{H}$, we use $\mu_1(\Sigma)\ge\mu_2(\Sigma)\ge \cdots$ to denote the
eigenvalues of $\Sigma$ in descending order, and we denote the
operator norm of $\Sigma$ by $\|\Sigma\|$. We use $I$ to denote
the identity operator on $\bbH$ and $I_n$ to denote the $n\times n$
identity matrix.
For an event $E$, $\mathbbm{1}\{E\}$ is an indicator function, which is $1$ if $E$ is true, and $0$ if $E$ is false.
For a sequence $\{a_n\}_n$, $O(a_n)$ and $o(a_n)$ denote Landau's big and small o notation, and $O_\Prob(a_n)$ and $o_\Prob(a_n)$ denote its probabilistic version.
We write $a_n = \Omega (b_n)$ for $\limsup_{n \to \infty} |a_n/b_n| > 0$, and $a_n = \omega (b_n)$ for $\limsup_{n \to \infty} |a_n/b_n| = \infty$
Also, $a_n = \Theta(b_n)$ means that both of $a_n = O(b_n)$ and $a_n = \Omega(b_n)$ hold.

\section{Setting: Linear Regression and Prediction for CATE}
\label{sec:linear}

\subsection{Conditional Average Treatment Effect}
We introduce the notion of conditoinal average treatment effect (CATE).
Suppose that there are two treatments $a \in\{1, 0\}$, where treatment $a=1$ corresponds to the active treatment, and treatment $a=0$ corresponds to the control treatment. 
We have access to $n$ training examples $\{(x_i, d_i, y_{i})\}^n_{i=1}$, each of which consists of a covariate $x_i$ from $\mathbb{H}$, a treatment indicator $d_i\in\{1, 0\}$, and a real-valued response variable $y_i$, and the examples are independent and identical copies of a jointly random element $(x,d,y_1, y_0) \sim P$ from the population distribution $P$.
$p(d = a| x )$ denotes the probability that an individual with covariate $x$ receives treatment $a\in\{1,0\}$, which is called the propensity score. 
We then posit the existence of potential outcome random variables $y_{1}$ and $y_{0}$ corresponding respectively to the
response subject would have experienced with and without treatment. By using the potential outcomes, the response variable $y$ can be written as
  \[y = \mathbbm{1}[d = 1]y_{1} + \mathbbm{1}[d = 0]y_{0}.\] 
We define the CATE at $x$ as
\begin{align*}
    \tau^*(x) = \mathbb{E}[y_{1}| x] - \mathbb{E}[y_{0}| x] 
\end{align*}
The main difficulty is that we can only observe a realized one of the two potential outcomes $y_{a}$ if $d=a$ for $a \in \{0,1\}$ from the given training example; thus, the model cannot be directly learned by minimizing the corresponding empirical risk.

\subsection{Linear Regression Model for Potential Outcomes}
\label{sec:liner_regression_model}
In this study, we consider a linear model for the potential outcome variables $y_a, a \in \{0,1\}$, then study the prediction of $\tau^*(x)$. 
For each $a \in \{0,1\}$, suppose that there exists a true parameter vector $\theta_a^* \in \mathbbm{H}$ such that the variables $x$ and $y_a$ follows the following linear model
\begin{align*}
    y_a = (\theta_a^*)^\top x + \varepsilon_a,
\end{align*}
where $\varepsilon_a$ is an independent noise variable whose mean is zero and variance is finite.
Note that $\varepsilon_0$ and $\varepsilon_1$ are independent to each other.
Note that $\theta_a^*$ is the optimal parameter for predicting $y_a$ without the selection, that is, $\Expect(y_a-x^\top\theta^*_a)^2=\min_{\theta \in \mathbb{H}}
      \Expect(y_a-x^\top\theta)^2$ holds.


For each $a\in\{0,1\}$, we give the following notations.
We define a {population covariance operator} $\Sigma=\Expect[x x^\top]$, and a {covariance operator with the selection assignment} $\Sigma_a=\Expect[\mathbbm{1}[d = a]x x^\top] = \Expect[p(d=a| x)x x^\top]$.
We consider an empirical potential outcome $y_{a,i}$ corresponding $y_i$ for $i=1,...,n$, and also define an empirical vector $\bs{y}_a\in\Re^n$ whose $i$-th entry is $\mathbbm{1}[d_i = a]y_i$, and a vector $\bs{\varepsilon}_a\in\Re^n$ whose $i$-th entry $\varepsilon_{a,i}=\mathbbm{1}[d_i = a]\varepsilon_{a,i}$, where $\varepsilon_{a,i}= (y_{a,i} - (\theta_a^*)^\top x_i)$. 
Further, following \cite{Bartlett2020}, we use the infinite matrix notation: $X_a$ denotes a linear map from $\bbH$ to $\Re^n$ corresponding to
  $(\mathbbm{1}[d_1 = a]x_1^\top,\ldots,\mathbbm{1}[d_n = a]x_n^\top)^\top$, so that $X_a\theta\in\Re^n$ has its $i$-th component as
  $\mathbbm{1}[d_i = a]x_{i}^\top\theta$ for $\theta \in \mathbb{H}$.
  We use the similar notation for the
  linear map $X^\top_a$ from $\Re^n$ to $\bbH$. Let us also define the linear map $X=X_1 + X_0$ from $\bbH$ to $\Re^n$; that is, it corresponds to
  $x_1,\ldots,x_n\in\bbH$. We also define $X^\top = X_1^\top + X_0^\top$.

\subsection{Excess Prediction Risk for CATE}


Given the observations $\{(x_i, d_i, y_{i})\}^n_{i=1}$, we consider an {\em estimator} which returns parameter ${\theta}\in\bbH$. Then, we \textit{predict} the CATE given covariate $x\in\mathbb{H}$ as $x^\top\theta$.
We measure the predictive performance of this estimator by using the excess risk.

  \begin{definition}[Excess risk]
  \label{def:excess}
  The excess risk of the estimator is defined as
    \[
      R\big(\theta\big) : = \mathbb{E}_{x,y}\Big[\big(\tilde{y} - x^\top \theta \big)^2 - \big(\tilde{y} - x^\top \theta^* \big)^2\Big],
    \]
where $\tilde{y} = y_{1} - y_{0}$, $\theta^* = \theta^*_1 - \theta^*_0$, and $\Expect_{x,y}$ denotes the conditional expectation given
  all random quantities other than $x,y_1, y_0$ (in this case, given the
  estimate $\theta$).
\end{definition}
This risk is used to predict the difference between the two groups to differentiate them from the treatment.
In the ordinary regression case, such as the excess risk defined in \citet{Bartlett2020}, the differences, $y_{1} - y_{0}$ and $\theta^*_1 - \theta^*_0$, are not included.

\subsection{Assumption}

To study the prediction risk, we need to make regular assumptions. 

\begin{assumption}[Basic] \label{asmp:basic}
   The following conditions hold:
   \begin{description}
   \setlength{\parskip}{0cm}
  \setlength{\itemsep}{0cm}
   	\item[1]\label{assumption:meanzero}
    $x$ and $y_a$ are mean zero for $a \in \{0,1\}$;
    \item[2]\label{assumption:sub-Gaussiandata}
    there exists an $\mathbb{H}$-valued random element $z$ which is conditionally $\sigma_x^2$-sub-Gaussian with $\sigma_x$, which means $\E[\exp(\lambda^\top z)] \leq \exp(\sigma_x^2\|\lambda\|^2/2)$ for all $\lambda \in \bbH$, and it satisfies $x = \Sigma^{1/2} z$. 
   	\item[3]
   	For $a \in \{0,1\}$, $\varepsilon_a$ is sub-Gaussian and has positive conditional variance, that is, there exist $\sigma_y^2, \sigma^2 > 0$ such that $\E[\exp(\lambda\varepsilon_a)|x] \leq \exp(\sigma_y^2\lambda^2/2)$ for any $\lambda \in \mathbbm{R}$ and $\Expect[\varepsilon_a^2|x] \geq \sigma^2$,
    \item[4] \label{assumption:full_rank}
  almost surely, for each $a \in \{0,1\}$, the projection of the
  data $X_a$ on the space orthogonal to any eigenvector of $\Sigma_a$
  spans a space of dimension $n$.
   	\end{description}
  \end{assumption}	
Assumption~\ref{asmp:basic} is a common setting in linear regression.
For example, Assumption~\ref{asmp:basic} is satisfied if $x$ and $y_a$ are jointly Gaussian with zero mean and $\rank(\Sigma)>n$ \citep{Bartlett2020}.

\begin{assumption}[Unconfoundedness] \label{asmp:unconfounded}
The treatment assignment $d$ is independent of the potential outcomes for $\{y_{1}, y_{0}\}$ conditional on $x$:
\[\{y_{1}, y_{0}\}\indep ~d \ |\ x.\]
\end{assumption}
Assumption \ref{asmp:unconfounded} expresses a natural setting wherein the assignment is independent of the output conditioned on the covariates.
This is the standard approach in treatment effect prediction \citep{Rosenbaum1983}.

\begin{assumption}[Overlap of assignment support] \label{asmp:coherent}
For some $0 < \varphi < 0.5$ and all $x\in\mathbb{H}$, 
\begin{align*}
    \varphi < p(d = 1| x ) < 1 - \varphi.
\end{align*}
\end{assumption}
Assumption~\ref{asmp:coherent} allows us to avoid overlap in treatment assignments.
The situation where no covariates are selected at all can also be avoided.
By this assumption, we can guarantee the identifiability of the assignment and the parameters \citep{imbens_rubin_2015}.

\section{Predictors with Interpolation}
We can construct specific method to predict the CATE $\tau^*(x)$ with linear models in several ways. In this study, we consider the following two prediction methods: the T-learner and IPW-learner.
\begin{description}
\item[The T-learner:] This method consists of a two-step procedure: in the first stage, we estimate the parameters of linear regression models for  $\mathbb{E}[y_1| x]$ and  $\mathbb{E}[y_0| x]$, separately; then, in the second stage, we predict the CATE by the difference of the two estimators.
\item[The IPW-learner:] This approach utilizes a propensity score $p(d=1|x)$ to constructs an conditionally unbiased estimator of $\tau^*(x_i)$ by using $y_i$ and $d_i$; 
then, constructs a predictor by regressing the unbiased estimator on the covariates $x$.
\end{description}

\subsection{The T-learner} 
In the first stage, we estimate $\theta_a^*$ for each $a \in \{0,1\}$ separately.
We consider an interpolating estimator, which can fit the data perfectly when the dimension of $x$ is larger than the sample size $n$; that is an interpolating estimator $\hat{\theta}_a$ satisfies $X_a\hat{\theta}_a=\bs{y}_a$ when $\mathbb{H} = \mathbb{R}^p $, where $p > n$. 
Specifically, we define the interpolating estimator for each $a \in \{0,1\}$ as
  \begin{align*}
    \hat\theta_a
      & =\arg\min_\theta\left\{\left\|\theta\right\|^2:
        X^\top_a X_a\theta = X^\top_a \bs{y}_a \right\}  = \left(X^\top_a X_a\right)^\dagger X^\top_a \bs{y}_a  = X^\top_a\left(X_aX^\top_a\right)^\dagger\bs{y}_a,
  \end{align*}
where $\left(X^\top_a X_a\right)^\dagger$ denotes the pseudoinverse of the
bounded linear operator $X^\top_a X_a$ (for infinite-dimensional $\bbH$,
the existence of the pseudoinverse is guaranteed because $X^\top_a X_a$ is
bounded and has a closed range; see~\cite{Desoer1963}). When $\bbH$ has dimension $p$ with $p<n$ and $X_a$ has rank $p$, there is a unique
solution to the normal equations. 
On the contrary,
the condition \ref{assumption:full_rank} in Assumption \ref{asmp:basic}
implies that we can find many solutions $\theta\in\bbH$ to the normal
equations that achieve $X_a\theta=y_a$.  
Hence, in this case, the minimum norm solution is given by
    $\hat\theta_a  = X^\top_a\left(X_aX^\top_a\right)^{-1}\bs{y}_a$.

In the second stage, we define a difference of the above estimators as $\hat{\theta}^{\mathrm{T\mathchar`-learner}} := \hat{\theta}_1 - \hat{\theta}_0$.
Then, we predict the CATE $\tau^*(x)$ by
\begin{align}
    \hat{\tau}^{\mathrm{T\mathchar`-learner}}(x) = x^\top \hat{\theta}^{\mathrm{T\mathchar`-learner}}= x^\top(\hat\theta_1 - \hat\theta_0). \label{def:predictor}
\end{align}
This approach of taking the difference between separate estimators from different treatment groups is a commonly considered method in causal inference without overparameterization.  This approach is named "T"-learner after the Two separate models used.

\subsection{The IPW-learner with Known Propensity Score} 

We utilize an approach to correct shift of distributions of the observed and population covariates.
Suppose that the propensity score $p(d=1| x)$ is known, which follows the setting of \citet{Horvitz1952,Austin2011}, we define a corrected response variable as
\begin{align}
    \hat{y}_i = \frac{\mathbbm{1}[d_i = 1]y}{p(d=1| x_i)} - \frac{\mathbbm{1}[d_i = 0]y_i}{1 - p(d=1| x_i)}. \label{def:yhat}
\end{align}
This variable $\hat{y}_i$ has the following property:
\begin{lemma} \label{lem:unbiased_yhat}
It is an unbiased estimator for $\tau^*(x_i)$, that is,  $\mathbb{E}\left[ \hat{y}_i | x_i \right] = \tau^*(x_i)$.
\end{lemma}

Then, we consider a regression problem with the corrected variable $\hat{y}_i$.
We define a vector $\hat{\bm{y}} =(\hat{y}_1,...,\hat{y}_n)$.
Then, we consider regressing $\hat{\bm{y}}$ on $X = X_1 + X_0$. 
We define an interpolating estimator using $(x_i, \hat{y}_i)$ as
  \begin{align*}
    \hat\theta^{\mathrm{IPW\mathchar`-learner}}
      & =\arg\min_\theta\left\{\left\|\theta\right\|^2:
        X^\top X\theta = X^\top \hat{\bs{y}} \right\}  = \left(X^\top X\right)^\dagger X^\top \hat{\bs{y}}  = X^\top\left(XX^\top\right)^{-1}\hat{\bs{y}}.
  \end{align*}
Then, we predict $\tau^*(x)$ by
\begin{align}
    \hat{\tau}^{\mathrm{IPW\mathchar`-learner}}(x) = x^\top \hat\theta^{\mathrm{IPW\mathchar`-learner}}. \label{def:IPW}
\end{align}


\section{Excess Risk Bounds and Benign Overfitting in the T-learner}
\label{sec:excess}

\subsection{Effective Rank and Related Notion}
To develop an upper bound of the risk of the T-learner, we define a notion of the effective rank to describe the benign overfitting, by following \cite{Bartlett2020}.

\begin{definition}[Effective Ranks]\label{def:ranks}
For a covariance operator $\Sigma$,
define $\lambda_i=\mu_i(\Sigma)$ for $i=1,2,\ldots$.
If $\sum_{i=1}^\infty\lambda_i<\infty$ and $\lambda_{k+1}>0$
for $k\ge 0$, define
  \begin{align*}
    r_k(\Sigma) & = \frac{\sum_{i>k}\lambda_i}{\lambda_{k+1}}, \quad \mbox{and} \quad 
    R_k(\Sigma)  = \frac{\left(\sum_{i>k}\lambda_i\right)^2}
        {\sum_{i>k}\lambda_i^2}.
  \end{align*}
\end{definition}
The effective rank is a measure of the complexity of covariances, using a tail of eigenvalues of the covariance matrix. 
It is used in the analysis of random matrices, such as in \cite{Koltchinskii2017}, to represent the benign overfitting of ordinary linear regression. 

Further, with a given $n \in \mathbb{N}$ and a constant $b > 0$, we define the following number of eigenvalues of the covariance operator $\Sigma$ which has a sufficiently large volume:
   \[
     k^* = k_n^* = \min\left\{k\ge 0: r_k(\Sigma)\ge bn\right\},
   \]
where the minimum of the empty set is defined as $\infty$. 
This notation is a continuation of the one used in \cite{Bartlett2020}.
The constant $b$ will be specified in theorems.

Also, with a given $n \in \mathbb{N}, \delta \in (0,1)$ and a covariance operator $\Sigma$, we define the following terms
\begin{align}
    \mathcal{B}_{n,\delta}(\Sigma) = \left\|\Sigma\right\|\max\left\{\sqrt{\frac{r_0(\Sigma)}{n}}, \frac{r_0(\Sigma)}{n}, \frac{\log(\delta^{-1})}{n} \right\}, \mbox{~and~} \mathcal{V}_n(\Sigma) = \sigma^2\left(\frac{k^*}{n}
         + \frac{n}{R_{k^*}(\Sigma)}\right). \label{def:bias_variance}
\end{align}
These terms are used to describe the excess risk. 
In general, while $\mathcal{B}_{n,\delta}(\Sigma)$ corresponds to a bias, $\mathcal{V}_{n}(\Sigma)$ relates to the variance of the prediction. 
How the aforementioned describe the error is explained by the following theorem.

\subsection{Upper Bound}
We derive our main result on the excess risk of the predictor.
At the first step, we derive an upper bound of the risk, which is decomposed into bias and variance related terms.
For $a \in \{0,1\}$, we define $P_a = (X_a X^\top_a )^{-1} X_a$, and  a projection operator to a complement space of $X_a$ as 
$\Pi_a^\bot := I - X^\top_a(X_aX^\top_a)^{-1}X_a$.

\begin{lemma}[Basic decomposition]The following inequality holds: \label{lem:bv_simple}
\begin{align*}
R\big(\hat{\theta}^{\mathrm{T\mathchar`-learner}}\big)&\leq  \sum_{a \in \{0,1\}} \left( 2\theta^{*\top}_a B_a \theta^*_a + 2\bm{\varepsilon}^\top_a C_a \bm{\varepsilon}_a + 2\theta^{*\top}_a D_{a,1-a} \bm \varepsilon_{1-a} \right) - 2\theta^{*\top}_1B_{1,0}\theta^*_0   - 2\bm \varepsilon^\top_1F_{1,0}\bm \varepsilon_0,
\end{align*}
where $B_a = \Pi_a^\bot \Sigma \Pi_a^\bot$, $B_{1, 0} = \Pi_1^\bot \Sigma \Pi_0^\bot$, $C_a = P_a \Sigma P_a^\top$, $D_{a,1-a} = \Pi_a^\bot \Sigma P_{1-a}^\top$, and $F_{1,0} = P_1 \Sigma P_0$.
\end{lemma}
The terms $2\theta^{*\top}_a B_a \theta^*_a $ and $\bm{\varepsilon}^\top_a C_a \bm{\varepsilon}_a$ mean bias and variance of the error by $\hat{\theta}_a$ for each $a \in \{0,1\}$.
In contrast, the rest of the terms $\theta^{*\top}_a D_{a,1-a} \bm \varepsilon_{1-a}$, $\theta^{*\top}_1B_{1,0}\theta^*_0$, and $\bm \varepsilon^\top_1F_{1,0}\bm \varepsilon_0$ denotes the cross effects between $\hat{\theta}_1$ and $\hat{\theta}_0$.
These items are specific to our setting.

By using the decomposition and the effective rank and notations, we present the following upper bound on  the excess risk in the following theorem:
\begin{theorem}[Upper bounds for the T-leaner in overparameterized moodels]
\label{thm:main}
For any $\sigma_x$ there are  $b,c, c_1>1$ for which the following holds. Consider a linear regression problem from
Section~\ref{sec:liner_regression_model} and suppose that Assumption~ \ref{asmp:basic}, \ref{asmp:unconfounded}, and \ref{asmp:coherent} hold.
Suppose $\delta<1$ with $\log(1/\delta)<n/c$.
If $k^* < n/c_1$, then the excess risk (Definition \ref{def:excess}) of the predictor in \eqref{def:predictor} satisfies with probability at least $1-\delta$,
\begin{align*}
    R(\hat{\theta}^{\mathrm{T\mathchar`-learner}}) &\leq \sum_{a\in\{1,0\}}\left( c\|\theta^*_a\|^2\mathcal{B}_{n,\delta}(\Sigma_a)+ \left\|\Sigma - \zeta^*_a\Sigma_a\right\|\|\theta^*_a\|^2\right) + c \|\theta_1^*\| \|\theta_0^*\| \mathcal{B}_{n,\delta}(\Sigma) \\
    & \quad + c \log (\delta^{-1}) \left\{ \mathcal{V}_n(\Sigma) + (\|\theta_1^*\| + \|\theta_0^*\|)\sqrt{\mathcal{V}_n(\Sigma)} \right\},\nonumber
\end{align*}
   where $\zeta^*_a = \argmin{\zeta\in\mathbb{R}^+} \|\Sigma - \zeta \Sigma_a\|$. 
\end{theorem}

We obtain the following three implications from this result. 
(i) Regarding the bias terms, all the covariances $\Sigma, \Sigma_1$, and $\Sigma_2$ affect the upper bound. 
This means that if, for example, the treatment assignment significantly increases the effective rank of the covariance $\Sigma_1$ for the particular treatment group, it prevents this upper bound from being reduced to zero. 
(ii) The variance terms depend only on the covariance $\Sigma$ of the population distribution and are unaffected by the covariance of each treatment group. This result contrasts with that of bias. 
(iii) Importantly, deviations of the covariance $\left\|\Sigma - \zeta^*_a\Sigma_a\right\|$ appear in the bound. 
This means that if the covariance $\Sigma_a$ of each treatment group $a \in \{0,1\}$ has an eigenspace that is largely different from the covariance $\Sigma$ of the population distribution, a non-negligible error will occur. This item has important implications for predicting treatment effects.

\subsubsection{Difference from Ordinary Linear Regression in \texorpdfstring{\citet{Bartlett2020}}{TEXT}}

For the ordinary regression without treatment effects, \citet{Bartlett2020} shows that an interpolating estimator has the following upper bound on the excess risk:
\begin{align*}
     c  \|\theta^*\|^2 \mathcal{B}_{n,\delta}(\Sigma) + c \log (\delta^{-1}) \mathcal{V}_n(\Sigma),
\end{align*}
where $\theta^* \in \mathbb{H}$ is a true parameter for the linear regression model.
Our derived upper bound differs from that of \citet{Bartlett2020} in several aspects. First, there is a deviation of the covariance, $\left\|\Sigma - \zeta^*_a\Sigma_a\right\|$. This term is due to the change in the distribution by the treatment assignment. Second, because of the cross terms from the different control groups; that is, $\hat{\theta}_1$, $\hat{\theta}_0$, $\epsilon_1$, and $\epsilon_0$, the terms $c\left\|\theta^*_1\right\|\left\|\theta^*_0\right\|\mathcal{B}_{n,\delta}(\Sigma)$ and $c(\left\|\theta^{*}_1\right\| + \left\|\theta^{*}_0\right\|)\log(1/\delta)\sqrt{\mathcal{V}_n(\Sigma)}$ appear. 
Especially, 
the first difference is critical because it restricts the problem instances where benign overfitting occurs.

This difference indicates that, unlike the ordinary linear regression, the CATE prediction requires additional conditions in the overparameterization scheme.
In the next section, we discuss the conditions for benign overfitting.

\begin{remark}[Excess risk lower bounds under overparameterized models]
For the readers' reference, we also remark the lower bound of the excess risk.
For any $\sigma_x$ there are  $b,c, c_1>1$ for which the following holds. Consider a linear regression problem from
Section~\ref{sec:liner_regression_model} and suppose that Assumption~\ref{asmp:basic}, \ref{asmp:unconfounded}, and \ref{asmp:coherent} hold.
 If $k^* \geq n/c_1$, then
 $\Expect R(\hat\theta)\ge \sigma^2/c$.
 Otherwise,
\begin{align*}
    \Expect R(\hat\theta^{\mathrm{T\mathchar`-learner}})\ge - c \|\theta_1^*\| \|\theta_0^*\| \bar{\mathcal{B}}_n(\Sigma) + c^{-1} \mathcal{V}_n(\Sigma),
\end{align*}
where $\bar{\mathcal{B}}_n(\Sigma) = \int_0^1 {\mathcal{B}}_{n,\delta}(\Sigma)\mathrm{d}\delta$.
Unlike the lower bound shown in \citet{Bartlett2020}, this lower bound does not give much information owing to the negative term $- c \|\theta_1^*\| \|\theta_0^*\| \bar{\mathcal{B}}_n(\Sigma)$. 
\end{remark}

\subsection{Conditions for Benign Overfitting}
\label{sec:cond}

We discuss the conditions for the convergence of the derived upper bound to zero, which is referred to as benign overfitting.
For $\Sigma$, we consider the following two cases. In the first case, $\Sigma$ is a fixed operator between infinite-dimensional Hilbert
spaces. 
In the second case, $\Sigma$ can change with $n$. For the latter case, we denote $\Sigma$ by $\Sigma_n$ to represent the dependency (resp. $\Sigma_a$ by $\Sigma_{a,n}$ for $a\in\{0,1\}$). Without loss of generality, in the following discussion, we may assume tbat $\|\Sigma\| = 1$. In the latter case, for an ordinary linear regression case without treatment effects, \citet{Bartlett2020} reveals that benign overfitting occurs when a sequence of covariance operators $\Sigma_n$ satisfies
 \begin{align}
\label{eq:benign_cov}
\lim_{n \rightarrow \infty} \frac{r_0(\Sigma_n)}{n}
 = \lim_{n \rightarrow \infty} \frac{k_n^*}{n}
 = \lim_{n \rightarrow \infty} \frac{n}{R_{k_n^*}(\Sigma_n)}
 = 0.
\end{align}
\citet{Bartlett2020} further analyzes the conditions for the upper bound to go to zero in their Theorem~2.
If \eqref{eq:benign_cov} holds, both $\mathcal{B}_{n,\delta}(\Sigma)$ and $\mathcal{V}_n(\Sigma)$ in \eqref{def:bias_variance} obviously converge to zero as $\delta \to 1$.

To describe the situation where the benign overfitting occurs, we summarize desirable properties of covariance operators $\Sigma$ as follow.
Here, we consider both the case (a) where $\Sigma$ is independent of $n$, and the case (b) where $\Sigma = \Sigma_n$ depends on $n$ and the dimension of $\Sigma_n$ increases with $n$.
These properties are developed by \cite{Bartlett2020}.
\begin{definition}[Benign covariance] \label{def:benign_cov}
We call a covariance matrix $\Sigma$ is \textit{benign}, if it satisfies either of the followings:
\begin{enumerate}
   \setlength{\parskip}{0cm}
  \setlength{\itemsep}{0cm}
    \item[(a)] $\Sigma$ is a fixed operator between infinite-dimensional spaces, and for some $\beta > 0$ it satisfies
    \begin{align*}
    &\mu_k(\Sigma) = k^{-\alpha} \ln^{-\beta} (k+1),
 \end{align*}
    \item[(b)] $\Sigma = \Sigma_n$ depends on $n$, and satisfies the following with $\gamma_k=\Theta(\exp(-k/\tau))$ for all $k \in \N$:
    \begin{align*}
        \mu_k(\Sigma_n) =( \gamma_k + \epsilon_n) \cdot \mathbbm{1}\{ k \leq p_n\}.
    \end{align*}
\end{enumerate}
\end{definition}
This condition concerns the covariance $\Sigma$ of the population distribution, which also evaluates the covariance $\Sigma_a$ of each treatment group $a \in \{0,1\}$.
\begin{lemma} \label{lem:benign_cov_a}
If $\Sigma$ is a benign covariance, we have $\max_{a \in \{0,1\}}r_0 (\Sigma_a) = o(n)$. 
\end{lemma}

In CATE prediction, in addition to the abovementioned condition, we need conditions:
\begin{assumption}[Coherent covariance] \label{asmp:coherent_eigen}
The following equality holds:
\begin{align}
\label{eq:main:thm}
    &\min_{\zeta\in\mathbb{R}^+} \left\|\Sigma - \zeta\Sigma_1\right\| = \min_{\zeta\in\mathbb{R}^+} \left\|\Sigma - \zeta\Sigma_0\right\| = 0
\end{align}
\end{assumption}
This condition is satisfied when the eigenspace of the covariance $\Sigma_a$ of each assigned group coincides with the eigenspace of the covariance $\Sigma$ of the population distribution. 
In other words, intuitively, the assignment rule must not change the data structure of the observation.

Then, combining Theorem~2 of \citet{Bartlett2020} and our new conditions \eqref{eq:main:thm}, we summarize the conditions for the benign overfitting in the following theorem.
\begin{theorem}[Benign overfitting in the T-learner]\label{theorem:benign_eigenvalues}
    Suppose assumptions in Theorem \ref{thm:main} holds.
    Also, suppose that $\Sigma$ is a benign covariance as Definition \ref{def:benign_cov}, and Assumption \ref{asmp:coherent_eigen} holds. 
    Then, 
    \begin{align*}
        R(\hat{\theta}^\mathrm{T\mathchar`-learner}) = o_\Prob(1), ~ (n \to \infty),
    \end{align*}
    if and only if $\alpha = 1$ and $ \beta > 1$ (case (a) in Definition \ref{def:benign_cov}), or $p_n=\omega(n)$ and
    $ne^{-o(n)}=\epsilon_np_n=o(n)$ (case (b) in Definition \ref{def:benign_cov}).
    
    Further, if we have $p_n=\Omega(n)$ and
    $\epsilon_np_n=ne^{-o(n)}$ in the case (b) in Definition \ref{def:benign_cov}, we obtain
    \[
        R(\hat\theta^{\mathrm{T\mathchar`-learner}})
          \!=\! O_\Prob\left(\frac{\epsilon_np_n+1}{n}
          + \frac{\ln(n/(\epsilon_np_n))}{n}
          + \max\left\{\frac{1}{n},\frac{n}{p_n}\right\}\right).
      \]
\end{theorem}
Except for the conditions~\eqref{eq:main:thm}, the statements are the identical to Theorem~2 in \citet{Bartlett2020}. Therefore, in CATE prediction, it is important to determine whether \eqref{eq:main:thm} is satisfied or not. 
\eqref{eq:main:thm} is more problematic because the term does not usually go to zero even if we put some conditions on the eigenvalues. To consider which cases satisfy \eqref{eq:main:thm}, we divide the problem instances into two situations: cases under RCTs and selection bias.

\subsection{Benign Overfitting in Randomized Control and Sample Selection}

\paragraph{Randomized Control Trial (RCT)}:
When we consider RCTs; tht is, when a treatment assignment does not depend on covariates, we can simplify the covariance $\Sigma_a$ for each treatment group simpler.
If the variable $d$ for the treatment assignment does not depend on $x$, then for each $a\in\{1,0\}$, we have
\begin{align*}
    \Sigma_a = \Expect[\mathbbm{1}[d = a]x x^\top] = p(d=a)\Expect[x x^\top] = p(d = a)\Sigma.
\end{align*}
Then, for $a \in \{0,1\}$, we have $\min_{\zeta\in\mathbb{R}^+} \|\Sigma - \zeta\Sigma_a\| = \|\Sigma - \frac{1}{p(d=a)}\Sigma_a\| = 0$ holds. 
This is a situation that standard RCTs satisfy, whereby we randomly assign treatments independently from the covariates. In addition, we also have $r_0(\Sigma_a) = p(d=a) r_0 (\Sigma)$.
Hence, the following statement holds without a proof.
\begin{proposition} \label{prop:RCT}
    Suppose that the propensity score does not depend on the covariates; that is $p(d = a|x) = p(d=a)$.
    Then, Assumption \ref{asmp:coherent} holds.
\end{proposition}
Therefore, when predicting the CATE on data collected by RCTs, the excess prediction risk converges to zero under the same conditions for linear regression discussed in \citet{Bartlett2020}.

\paragraph{Selection Bias}: If there is a selection bias in the sense that the treatment assignment depends on the covariate $x$, the conditions under which benign overfitting occurs are severe. 
Under selection bias, the covariance $\Sigma_a$ for each treatment group $a \in \{0,1\}$ is given as 
    $\Sigma_a = \Expect_x[\Expect[p(d=a|x)x x^\top| x]]$.
This disallows us from decomposing $\Sigma_a$ into $p(d=a| x)$ and $\Sigma$. Therefore, we cannot make the deviation $\min_{\zeta\in\mathbb{R}^+} \left\|\Sigma - \zeta\Sigma_a\right\|$ zero.

\section{Excess Risk Bounds and Benign Overfitting in the IPW-learner}
\label{sec:excess_risk_ipw}

First, to evaluate the IPW-learner, we define a new excess risk as
\[
      \tilde{R}\big(\theta\big) : = \mathbb{E}_{x,y}\Big[\big(\hat{y} - x^\top \theta \big)^2 - \big(\hat{y} - x^\top \theta^* \big)^2\Big].
    \]
Importantly, this new risk is same to the risk given in Definition \ref{def:excess}.
\begin{lemma} \label{lem:equiv_risks}
For any $\theta \in \mathbbm{H}$, we have $R(\theta) = \tilde{R}({\theta})$.
\end{lemma}

Next, when using the IPW estimator, we confirm that Assumptions~1, 3, and 4 in Definition~1 of \citet{Bartlett2020} can be replaced as follows.
\begin{lemma}[Basic assumptions in the IPW-learner] \label{lem:basic}
   Under Assumptions~\ref{asmp:basic}--\ref{asmp:coherent}, the following hold:
   \begin{description}
      \setlength{\parskip}{0cm}
  \setlength{\itemsep}{0cm}
   	\item[1']
    $x$ and $\hat{y}$ are mean zero;
   	\item[2']
    the {\em conditional noise variance} is bounded below by some constant $\tilde{\sigma}^2$: $\Expect[(\hat{y}-x^\top\theta^*)^2|x] \geq \tilde{\sigma}^2$.
   	\item[3']
    $\hat{y} - x^\top\theta^*$ is $\tilde{\sigma}_y^2$-sub-Gaussian, conditionally on $x$, that is for all $\lambda \in \Re$, $\E[\exp(\lambda(\hat{y} - x^\top\theta^*))|x] \leq \exp(\tilde{\sigma}_y^2\lambda^2/2)$.
   	\end{description}
  \end{lemma}
Since the statements in Lemma~\ref{lem:basic} correspond to  Assumptions~1, 3, and 4 in Definition~1 in \citet{Bartlett2020}, we can directly apply Theorem~1 of \citet{Bartlett2020} to obtain the following result by combining them with Assumptions~\ref{asmp:basic}.
\begin{theorem}[Excess risk upper bounds in IPW-learner]
\label{thm:main2}
For any $\sigma_x$ there are  $b,c, c_1>1$ for which the following holds. Consider a linear regression problem from Section~\ref{sec:liner_regression_model} and suppose that Assumption \ref{asmp:basic}, \ref{asmp:unconfounded}, and \ref{asmp:coherent}, hold.
Suppose $\delta<1$ with $\log(1/\delta)<n/c$.
If $k^* < n/c_1$, then the excess risk (Definition \ref{def:excess}) of the predictor in \eqref{def:IPW} satisfies
   \begin{align*}
     R\big(\hat{\theta}^{\mathrm{IPW\mathchar`-learner}}\big)
       & \le c\|\theta^*\|^2\mathcal{B}_{n,\delta}(\Sigma)+ c\log(1/\delta)\mathcal{V}_n(\Sigma).
   \end{align*}
 with probability at least $1-\delta$.
\end{theorem}
Thus, the upper bound has the same form as that of \citet{Bartlett2020}, although some constant terms are affected by the construction of the IPW estimator. As presented in \citet{Bartlett2020}, if \eqref{eq:benign_cov} is satisfied, the upper bound goes to zero. This means that this upper bound goes to zero under the same condition as that in \citet{Bartlett2020}.

\begin{theorem}[Benign overfitting in IPW-learner]\label{prp:benign_eigenvalues}
    Suppose that the assumptions in Theorem \ref{thm:main2} hold.
    Also, suppose that $\Sigma$ is a benign covariance as Definition~\ref{def:benign_cov}. 
    Then, we have
    \begin{align*}
        R(\hat{\theta}^\mathrm{IPW\mathchar`-learner}) = o_\Prob(1), ~ (n \to \infty),
    \end{align*}
    if and only if $\alpha = 1$ and $ \beta > 1$ (case (a) in Definition \ref{def:benign_cov}), or $p_n=\omega(n)$ and
    $ne^{-o(n)}=\epsilon_np_n=o(n)$ (case (b) in Definition \ref{def:benign_cov}).
    
    Further, if we have $p_n=\Omega(n)$ and
    $\epsilon_np_n=ne^{-o(n)}$ the case (b) in Definition \ref{def:benign_cov}), we obtain
    \[
        R(\hat\theta^{\mathrm{IPW\mathchar`-learner}})
          \!=\! O_\Prob\left(\frac{\epsilon_np_n+1}{n}
          + \frac{\ln(n/(\epsilon_np_n))}{n}
          + \max\left\{\frac{1}{n},\frac{n}{p_n}\right\}\right).
      \]
\end{theorem}
In this result, unlike the case where we use the T-learner, the upper bound goes to zero for both cases with RCTs and selection bias. This is because the deviation term $\|\Sigma - \zeta^*\Sigma_a\|$ in Theorem~\ref{thm:main} does not appear when using the IPW-learner.

\section{Discussion}
\label{sec:discuss}
\subsection{CATE Estimation in Non-overparametrized Setting}
When the model is not overparametrized, both the T-learner and IPW-learner have a risk converges to zero as $n \to \infty$ \citep{Abrevaya2015}.
This fact is contrasive to our results with overparameterization, especially for the case of the T-leaner.
This result implies the potential danger of using the T-learner in large-scale models.


\subsection{Application to Doubly Robust Estimator}
The doubly robust estimator \citep{Porter2011,Funk2011,Foster2019,Kennedy2020} is another common choice for the CATE prediction.
An advantage of the double robust estimator is that an asymptotic variance of the estimator is semiparametric efficient; that is, the asymptotic variance achieves its lower bound \citep{hahn1998role,ChernozhukovVictor2018Dmlf}. 
However, some problems appear in the overparameterization setting, when this method estimates an conditional outcome $\mathbb{E}[y_{a,i}| x_i]$ as preparation.
Specifically, according to our results, the T-learner is not guaranteed to be valid to estimate the outcome under overparameterization.
Even if the IPW-learner is used, we need the correct propensity score.
Unless these issues are resolved, it is difficult to utilize the doubly robust estimator with overparameterization.


\subsection{Implications for Applications}
Our result implies the importance of bias correction methods.
When data is collected via RCTs, we recommend the usual method, such as performing regressions for each assigned group like the T-leaner.
In contrast, if there is selection bias, we recommend to achieve a correct assigned probability and utilize the IPW-learner.

\section{Conclusion}
We investigate the necessary conditions for benign overfitting in CATE prediction using a linear regression model. 
For predicting the CATE, we consider two methods: the T-learner and IPW learner. 
When the treatment assignment $d$ does not depend on the covariates, both the T-learner and IPW-learner with interpolating prediction rule show benign-overfitting. 
However, when the treatment assignment $d_i$ depends on the covariates, the excess risk of the T-learner does not converge to zero, while that of the IPW-learner converges. Thus, this paper shows the situation in which CATE prediction with an overparameterized model works.


\bibliographystyle{asa}
\bibliography{arXiv.bbl}

\clearpage

\appendix

\section{Auxiliary Results from \texorpdfstring{\citet{Bartlett2020}}{TEXT} and \texorpdfstring{\citet{vershynin_2018}} {TEXT}}

\begin{proposition}[Lemma~S.2 in \citet{Bartlett2020}]
	\label{lemma::sub-Gaussian-quadratic-form}
  Consider random variables $\eps_1,\ldots,\eps_n$, conditionally
  independent given $X$ and conditionally $\sigma^2$-sub-Gaussian, that
  is, for all $\lambda\in\Re$,
    \[
      \E[\exp(\lambda\eps_i)|X] \leq \exp(\sigma^2\lambda^2/2).
    \]
  Suppose that, given $X$, $M\in\Re^{n\times n}$ is a.s.~positive
  semidefinite. Then a.s.~on $X$, with conditional probability at
  least $1-\exp(-t)$,
    \[
      \bs{\eps}^\top M \bs{\eps}
        \le \sigma^2\tr(M) + 2\sigma^2 \|M\|t
          + 2\sigma^2\sqrt{\|M\|^2 t^2 +
          \tr\left(M^2\right)t}.
    \]
\end{proposition}

\begin{proposition}[Lemma~S.3 in \citet{Bartlett2020}]\label{lem:cor:s3}
    Suppose $k < n$, $A \in \R^{n\times n}$ is an invertible matrix,
  and $Z \in \R^{n\times k}$ is such that $ZZ^\top + A$ is invertible.
  Then
    \[
      Z^\top (ZZ^\top + A)^{-2} Z
      = (I +  Z^\top A^{-1}Z)^{-1}Z^\top A^{-2}Z
        (I +  Z^\top A^{-1}Z)^{-1}.
     \]
\end{proposition}

\begin{proposition}[Bernstein's inequality, Lemma~S.5 in \citet{Bartlett2020} 
]\label{lemma:Bernstein}
    There is a universal constant $c$ such that, for any
	independent, mean zero, $\sigma$-sub-exponential random variables
	$\xi_1, \dots, \xi_N$, any $a = (a_1, \dots, a_N)\in \Re^n$, and
    any $t \geq 0$,
	  \begin{align*}
	    \Pbb\left(\left|\sum_{i=1}^N a_i \xi_i\right| > t\right)
            \leq 2\exp\left(-c\min\left(\frac{t^2}{\sigma^2 \sum_{i=1}^N
        a_i^2}, \frac{t}{\sigma \max_{1 \leq i \leq n}
        a_i}\right)\right).
	  \end{align*}
\end{proposition}

\begin{proposition}[Corollary~S.6 in \citet{Bartlett2020}]\label{cor:Bernstein}
    There is a universal constant $c$ such that for any
	non-increasing sequence $\{\lambda_i\}_{i =1}^\infty $
    of non-negative numbers such that
    $\sum_{i=1}^\infty \lambda_i < \infty$, and any
    independent, centered, $\sigma$-sub-exponential random variables
    $\{\xi_i\}_{i = 1}^\infty$, and any $x > 0$, with probability
    at least $1-2e^{-x}$
	  \[
	    \left|\sum_{i} \lambda_i \xi_i\right| \leq
        c\sigma\max\left(x\lambda_1, \sqrt{x \sum_{i}
        \lambda_i^2}\right).
	  \]
\end{proposition}

\begin{proposition}[$\epsilon$-net argument, Lemma~S.8 in \citet{Bartlett2020}]\label{lemma::norm_by_net}
  Suppose $A \in \R^{n\times n}$ is a symmetric matrix, and
  $\Nc_\epsilon$ is an $\epsilon$-net on the unit sphere $\Sc^{n-1}$ in
  the Euclidean norm, where $\epsilon < \frac12$. Then
    \[
      \|A\| \leq (1 - \epsilon)^{-2} \max_{x \in \Nc_\epsilon} |x^\top Ax|.
    \]
\end{proposition}

\begin{proposition}[General Hoeffding inequality, Theorem 2.6.3 in \citet{vershynin_2018}]\label{lemma:Hoeffding}
    There is a universal constant $c$ such that, for any
	independent, mean zero, $\sigma$-sub-Gaussian random variables $\xi_1, \dots, \xi_N$, any $a = (a_1, \dots, a_N)\in \Re^n$, and
    any $t \geq 0$,
	  \begin{align*}
	    \Pbb\left(\left|\sum_{i=1}^N a_i \xi_i\right| > t\right)
            \leq 2\exp\left(-c\frac{t^2}{\sigma^2 \sum_{i=1}^N
        a_i^2}\right).
	  \end{align*}
\end{proposition}

\begin{corollary}\label{cor:Hoeffding}
    There is a universal constant $c$ such that for any
	non-increasing sequence $\{\lambda_i\}_{i =1}^\infty $
    of non-negative numbers such that
    $\sum_{i=1}^\infty \lambda_i < \infty$, and any
    independent, centered, $\sigma$-sub-Gaussian random variables
    $\{\xi_i\}_{i = 1}^\infty$, and any $x > 0$, with probability
    at least $1-2e^{-x}$
	  \[
	    \left|\sum_{i} \lambda_i \xi_i\right| \leq
        c\sigma\sqrt{x \sum_{i}
        \lambda_i^2}.
	  \]
\end{corollary}

\section{Proof of Lemma~\ref{lemma:bv}}
\label{appdx:lemma:bv}
\begin{proof}
First, as in the proof of \citet{Bartlett2020}, we have
\begin{align*}
R\big(\hat{\theta}^{\mathrm{T\mathchar`-learner}}\big)=\mathbb{E}_{x,y}\Big[\big(\tilde{y} - x^\top \hat{\theta} \big)^2 - \big(\tilde{y} - x^\top \theta^* \big)^2\Big] = \mathbb{E}_{x}\Big[\Big(x^\top \big( \theta^* - \hat{\theta} \big)\Big)^2\Big].
\end{align*}
We can decompose $\mathbb{E}_{x}\Big[\Big(x^\top \big( \theta^* - \hat{\theta} \big)\Big)^2\Big]$ as
\begin{align*}
&\mathbb{E}_{x}\Big[\Big(x^\top \big( \theta^* - \hat{\theta} \big)\Big)^2\Big]\\
&= \mathbb{E}_{x}\Big[\Big(x^\top \big(\theta^*_1 - \hat{\theta}_1\big) - x^\top\big( \theta^*_0 - \hat{\theta}_0 \big)\Big)^2\Big]\\
&= \annot{\mathbb{E}_{x}\Big[\Big(x^\top \big(\theta^*_1 - \hat{\theta}_1\big)\Big)^2\Big]}{Excess risk of $y_{1}$.} + \annot{\mathbb{E}_{x}\Big[\Big(x^\top\big( \theta^*_0 - \hat{\theta}_0 \big)\Big)^2\Big]}{Excess risk of $y_{0}$} - 2\annot{\mathbb{E}_{x}\Big[\Big( \big(\theta^*_1 - \hat{\theta}_1\big)^\top x\Big)\Big( x^\top\big(\theta^*_0 - \hat{\theta}_0 \big) \Big)\Big]}{Excess risk incurred by the difference between $y_{1}$ and $y_{0}$.}.
\end{align*}
In Lemma~2 of \citet{Bartlett2020}, by using Lemma~S.2 and S.18 of \citet{Bartlett2020}, the authors show 
\begin{align*}
\mathbb{E}_{x}\Big[\Big(x^\top \big(\theta^*_1 - \hat{\theta}_1\big)\Big)^2\Big] &\leq 2\theta^{*\top}_1 B_1 \theta^*_1 + 2\bm{\varepsilon}^\top_1 C_1 \bm{\varepsilon}_1,\\
\mathbb{E}_{x,\bm{\varepsilon}}\Big[\Big(x^\top \big(\theta^*_1 - \hat{\theta}_1\big)\Big)^2\Big]&\geq \theta^{*\top}_1 B_1 \theta^*_1 + \sigma^2_1 \mathrm{tr}(C_1),\\
\mathbb{E}_{x}\Big[\Big(x^\top\big( \theta^*_0 - \hat{\theta}_0 \big)\Big)^2\Big] &\leq 2\theta^{*\top}_0 B_0 \theta^*_0 + 2\bm{\varepsilon}^\top_0 C_0 \bm{\varepsilon}_0,\\
\mathbb{E}_{x,\bm{\varepsilon}}\Big[\Big(x^\top\big( \theta^*_0 - \hat{\theta}_0 \big)\Big)^2\Big] &\geq  \theta^{*\top}_0 B_0 \theta^*_0 + \sigma^2_0 \mathrm{tr}(C_0).
\end{align*}
Our remaining task is to consider the bound of $ \Big( \big(\theta^*_1 - \hat{\theta}_1\big)^\top x\Big)\Big( x^\top\big(\theta^*_0 - \hat{\theta}_0 \big) \Big)$, which is decomposed as
\begin{align*}
&\mathbb{E}_{x}\Big[\Big( \big(\theta^*_1 - \hat{\theta}_1\big)^\top x\Big)\Big(x^\top\big(\theta^*_0 - \hat{\theta}_0 \big) \Big)\Big]\\
&= \mathbb{E}_{x}\Big[\Big(\theta^{*\top}_1\big(I- X^\top_1(X_1X^\top_1)^{-1}X_1 \big) x - \bm \varepsilon^\top_1(X_1X^\top_1)^{-1}X_1  x\Big)\\
&\ \ \ \ \ \ \ \times \Big(x^\top \big(I - X^\top_0(X_0X^\top_0)^{-1}X_0 \big) \theta^*_0 - x^\top X^\top_0(X_0X^\top_0)^{-1}\bm \varepsilon_0\Big)\Big]\\
&= \theta^{*\top}_1\big(I - X^\top_1(X_1X^\top_1)^{-1}X_1 \big) \Sigma \big(I - X^\top_0(X_0X^\top_0)^{-1}X_0 \big) \theta^*_0 \\
&\ \ \ - \theta^{*\top}_1\big(I - X^\top_1(X_1X^\top_1)^{-1}X_1 \big) \Sigma X^\top_0(X_0X^\top_0)^{-1} \bm \varepsilon_0 \\
&\ \ \ - \bm \varepsilon^\top_1(X_1X^\top_1)^{-1}X_1 \Sigma \big(I - X^\top_0(X_0X^\top_0)^{-1}X_0 \big) \theta^*_0\\
&\ \ \ + \bm \varepsilon^\top_1(X_1X^\top_1)^{-1}X_1  \Sigma X^\top_0(X_0X^\top_0)^{-1} \bm \varepsilon_0.
\end{align*}
Therefore, $\mathbb{E}_{x}\Big[\Big( \big(\theta^*_1 - \hat{\theta}_1\big)^\top x\Big)\Big(x^\top\big(\theta^*_0 - \hat{\theta}_0 \big) \Big)\Big] = \theta^{*\top}_1B_{1,0}\theta^*_0 - \theta^{*\top}_1D \bm \varepsilon_0 - \bm \varepsilon^\top_1 E\theta^*_0 + \bm \varepsilon^\top_1F\bm \varepsilon_0$, and $\mathbb{E}_{x,\bm{\varepsilon}}\Big[\Big( \big(\theta^*_1 - \hat{\theta}_1\big)^\top x\Big)\Big(x^\top\big(\theta^*_0 - \hat{\theta}_0 \big) \Big)\Big] = \theta^{*\top}_1B_{1,0}\theta^*_0$.

Thus, we obtain the upper bound as
\begin{align*}
&R\big(\hat{\theta}^{\mathrm{T\mathchar`-learner}}\big)\\
&\leq 2\theta^{*\top}_1 B_1 \theta^*_1 + 2\bm{\varepsilon}^\top_1 C_1 \bm{\varepsilon}_1 + 2\theta^{*\top}_0 B_0 \theta^*_0 + 2\bm{\varepsilon}^\top_0 C_0 \bm{\varepsilon_0}- 2\theta^{*\top}_1B_{1,0}\theta^*_0 + 2\theta^{*\top}_1E \bm \varepsilon_0 + 2\bm \varepsilon^\top_1 F\theta^*_0 - 2\bm \varepsilon^\top_1G\bm \varepsilon_0
\end{align*}
and the lower bound as 
\begin{align*}
&\mathbb{E}_{x,\bm{\varepsilon}}\Big[R\big(\hat{\theta}^{\mathrm{T\mathchar`-learner}}\big) \Big] \geq \theta^{*\top}_1 B_1 \theta^*_1 + \sigma^2_1 \mathrm{tr}(C_1) + \theta^{*\top}_0 B_0 \theta^*_0 + \sigma^2_1 \mathrm{tr}(C_0) - 2\theta^{*\top}_1B_{1,0}\theta^*_0. 
\end{align*}
\end{proof}

\section{Proof of Lemma~\ref{lemma:bias_1_0}}\label{a:upperB}
\label{appdx:a:upperB}
\begin{proof}
We consider only the case $a=1$.
The case $a=0$ can be proved in the same way. Note that
  \begin{align}\label{equation:Xorthog}
    \left(I - X^\top_1 \left(X_1 X^\top_1 \right)^{-1}X_1\right)X^\top_1
    &= X^\top_1  - X^\top_1 \left(X_1 X^\top_1 \right)^{-1}(X_1X^\top_1 ) =0.
  \end{align}
Moreover, for any $v$ in the orthogonal complement to the span of
the columns of $X^\top_1$,
  \[
    \left(I - X^\top_1 \left(X_1 X^\top_1 \right)^{-1}X_1\right)v = v.
  \]
Thus, 
  \begin{equation}\label{eqn:projection}
    \|I - X^\top_1 \left(X_1 X^\top_1 \right)^{-1}X_1\|\leq 1.
  \end{equation}
Given $\zeta^* = \argmin{\zeta\in\mathbb{R}^+} \left\|\Sigma - \zeta\Sigma_a\right\|$, we apply~\eqref{equation:Xorthog} to write
  \begin{align*}
    {\theta^*_1}^\top B_1 \theta^*_1
    & 
     = {\theta^*_1}^\top \left(I - X^\top_1
      \left(X_1 X^\top_1 \right)^{-1}X_1\right)\Sigma\left(I - X^\top_1
      \left(X_1 X^\top_1 \right)^{-1}X_1\right)\theta^*_1 \\
    &= {\theta^*_1}^\top \left(I
      - X^\top_1 \left(X_1 X^\top_1 \right)^{-1}X_1\right)\left(\zeta^*\Sigma_1 -
      \zeta^*\frac{1}{n} X^\top_1 X_1\right) 
        \left(I - X^\top_1
      \left(X_1 X^\top_1 \right)^{-1}X_1\right)\theta^*_1\\
     &\ \ \ + {\theta^*_1}^\top \left(I
      - X^\top_1 \left(X_1 X^\top_1 \right)^{-1}X_1\right)\left(\Sigma -
      \zeta^*\Sigma_1\right) 
        \left(I - X^\top_1
      \left(X_1 X^\top_1 \right)^{-1}X_1\right)\theta^*_1.
  \end{align*}
Combining with~\eqref{eqn:projection} shows that 
  \[
    {\theta^*_1}^\top B_1 \theta^*_1
    \leq \zeta^*\left\|\Sigma_1 - \frac1n X^\top_1 X_1\right\|\|\theta^*_1\|^2 + \left\|\Sigma - \zeta^*\Sigma_1\right\|\|\theta^*_1\|^2.
  \]
Thus, due to Theorem~9 in~\cite{Koltchinskii2017},
there is an absolute constant $c$ such that for any $t > 1$  with
probability at least $1-\exp(-t)$,
  \[
    {\theta^*_1}^\top B_1 \theta^*_1
    \leq c\|\theta^*_1\|^2\|\Sigma_1\|\max\left\{\sqrt{\frac{r(\Sigma_1)}{n}},
      \frac{r(\Sigma_1)}{n}, \sqrt{\frac{t}{n}}, \frac{t}{n} \right\} + \left\|\Sigma - \zeta^* \Sigma_1\right\|\|\theta^*_1\|^2,
  \]
where 
  \[
    r(\Sigma_1) := \frac{(\E\|x\|)^2}{\|\Sigma_1\|}
    \leq \frac{\tr(\Sigma_1)}{\|\Sigma_1\|}
    = \frac{1}{\lambda_1}\sum_i \lambda_i = r_0(\Sigma_1).
  \]
Since $1<t<n$ implies $\sqrt{\frac{t}{n}} > \frac{t}{n}$, 
 \[
    {\theta^*_1}^\top B_1 \theta^*_1
    \leq c\|\theta^*_1\|^2\|\Sigma_1\|\max\left\{\sqrt{\frac{r(\Sigma_1)}{n}},
      \frac{r(\Sigma_1)}{n}, \sqrt{\frac{t}{n}} \right\} + \left\|\Sigma - \zeta^* \Sigma_1\right\|\|\theta^*_1\|^2,
  \]
\end{proof}

\section{Proof of Lemma~\ref{lemma:bias_10}}
\label{appdx:lemma:bias_10}
\begin{proof}
Recall that 
\begin{align*}
B_{1,0} = \big(I - X^\top_1(X_1X^\top_1)^{-1}X_1 \big) \Sigma \big(I - X^\top_0(X_0X^\top_0)^{-1}X_0 \big)
\end{align*}

Now we can apply~\eqref{equation:Xorthog} to write
  \begin{align*}
    &{\theta^*_1}^\top B_{1,0} \theta^*_0\\
    & 
     = {\theta^*_1}^\top \left(I - X^\top_1
      \left(X_1 X^\top_1 \right)^{-1}X_1\right)\Sigma\left(I - X^\top_0
      \left(X_0 X^\top_0 \right)^{-1}X_0\right)\theta^*_0 \\
    &= {\theta^*_1}^\top \left(I
      - X^\top_1 \left(X_1 X^\top_1 \right)^{-1}X_1\right)\left(\Sigma -
      \frac1n X^\top X\right) 
        \left(I - X^\top_0
      \left(X_0 X^\top_0 \right)^{-1}X_0\right)\theta^*_0 + {\theta^*_1}^\top\frac1n X^\top_0 X_1\theta^*_0,
  \end{align*}
  where $X^\top X = X^\top_1X_1 + X^\top_1X_0 + X^\top_0X_1 + X^\top_0X_0$.
Here, note that $X^\top_0 X_1=0$. Combining with~\eqref{eqn:projection} shows that
  \[
    {\theta^*_1}^\top B_{1,0} \theta^*_1
    \leq \left\|\Sigma - \frac1n X^\top X\right\|\|\theta^*_0\|\|\theta^*_1\|.
  \]
Thus, due to Theorem~9 in~\cite{Koltchinskii2017},
there is an absolute constant $c$ such that for any $t > 1$  with
probability at least $1-\exp(-t)$,
  \[
    \Big|2{\theta^*}^\top_1 B_{1,0} \theta^*_0\Big|
    \leq c\|\theta^*_1\|\|\theta^*_0\|\|\Sigma\|\max\left\{\sqrt{\frac{r(\Sigma)}{n}},
      \frac{r(\Sigma)}{n}, \sqrt{\frac{t}{n}}, \frac{t}{n} \right\},
  \]
where 
  \[
    r(\Sigma) := \frac{(\E\|x\|)^2}{\|\Sigma\|}
    \leq \frac{\tr(\Sigma)}{\|\Sigma\|}
    = \frac{1}{\lambda_1}\sum_i \lambda_i = r_0(\Sigma).
  \]
\end{proof}

\section{Proof of Lemma~\ref{lem:bound_d_e}}
\label{appdx:lem:bound_d_e}
We use the fact $X_1 \Sigma X_0^\top = 0$ and compute the term as
\begin{align*}
D &= \big(I - X^\top_1(X_1X^\top_1)^{-1}X_1 \big) \Sigma X^\top_0(X_0X^\top_0)^{-1} = \Sigma X^\top_0(X_0X^\top_0)^{-1}
\end{align*}

\begin{lemma}
\label{lem:concent_ineq1}
For each $a\in\big\{0, 1\big\}$, consider random variables $\bm{\varepsilon}_a = \big(\varepsilon_{a, 1}\ \cdots\ \varepsilon_{a, n}\big)$, 
conditionally independent given $X_a$ 
and conditionally $\sigma^2$ sub-Gaussian, that is, for all $\lambda \in \mathbb{R}$,
\begin{align*}
    \mathbb{E}\left[\exp\big(\lambda \varepsilon_{a, i}\big)| X_a\right] \leq \exp\big(\sigma^2\lambda^2/2\big).
\end{align*}
Then a.s. on $X_1$ and $X_0$, with conditional probability at least $1-\exp(-t)$,
\begin{align*}
   \theta^{*\top}_1\Sigma X^\top_0(X_0X^\top_0)^{-1}\bm{\varepsilon}_0 <  \sqrt{2t\sigma^2\left\|\theta^{*\top}_1\Sigma X^\top_0(X_0X^\top_0)^{-1}\right\|^2_2}.
\end{align*}
Similarly, a.s. on $X_1$ and $X_0$, with conditional probability at least $1-\exp(-t)$,
\begin{align*}
   \bm{\varepsilon}_1(X_1X^\top_1)^{-1}X_1 \Sigma\theta^{*}_0 <  \sqrt{2t\sigma^2\left\|\theta^{*\top}_0\Sigma X^\top_1(X_1X^\top_1)^{-1}\right\|^2_2}.
\end{align*}
\end{lemma}
\begin{proof}
We only prove the first statement.
The second statement can be shown in the same way.

For any $ \bs v = (v_1,...,v_n) \in \mathbb{R}_+^n$, the following inequality holds: 
\begin{align*}
 &\mathbb{E}\left[\exp\left(\bm{v}^\top\bm{\varepsilon}_a\right) \mid X_a\right] = \prod^n_{i=1}\mathbb{E}\left[\exp\left(v_i\varepsilon_{a,i}\right) \mid X_a\right]\\
 & \leq \prod^n_{i=1}\exp\left(\sigma^2v^2_i/2\right) = \exp\left(\sigma^2\sum^n_{i=1}v^2_i/2\right) = \exp\left(\sigma^2\left\|\bm{v}\right\|^2_2/2\right).
\end{align*}
By replacing $\bm{v}$ with $\lambda \theta^{*\top}_1\Sigma X^\top_0(X_0X^\top_0)^{-1}$ and using Chernoff bound \citep[Section~2.2,][]{Boucheron2004}, for any $s \geq 0$,
\begin{align*}
     \theta^{*\top}_1\Sigma X^\top_0(X_0X^\top_0)^{-1}\bm{\varepsilon} > s
\end{align*}
for $\lambda \geq 0$ with probability at most %
\begin{align*}
\frac{\mathbb{E}\left[\exp\left(\bm{v}^\top\bm{\varepsilon}_a\right) \mid X_a\right]}{\exp(-\lambda s)} \leq \exp\left(\lambda^2\sigma^2\left\|\theta^{*\top}_1\Sigma X^\top_0(X_0X^\top_0)^{-1}\right\|^2_2/2 - \lambda s\right)
\end{align*}
almost surely conditional on $X_1$ and $X_0$. Letting
\begin{align*}
    \lambda = \frac{s}{\sigma^2\left\|\theta^{*\top}_1\Sigma X^\top_0(X_0X^\top_0)^{-1}\right\|^2_2}
\end{align*}
gives the bound
\begin{align*}
    \exp\left(-\frac{s^2}{2\sigma^2\left\|\theta^{*\top}_1\Sigma X^\top_0(X_0X^\top_0)^{-1}\right\|^2_2}\right).
\end{align*}
Then, letting $t = \frac{s^2}{2\sigma^2\left\|\theta^{*\top}_1\Sigma X^\top_0(X_0X^\top_0)^{-1}\right\|^2_2}$, we complete the proof.
\end{proof}

\section{Related Results in Section~\ref{sec:upper_trace}}
\label{appdx:sec:upper_trace}

\begin{corollary}[Refinement of Corollary~S.7 in \citet{Bartlett2020}]
	\label{cor::cor norm of projection}
	Suppose $z\in \Re^n$ is a random vector with independent
    $\sigma^2$-sub-Gaussian coordinates with unit variances,
    $\mathcal{L}$ is a random subspace of $\mathbb{R}^n$ of codimension
    $k$, and $\mathcal{L}$ is independent of $z$. Then for some universal
    constants $c_1, c_2$ and any $t > 0$, with probability at least
    $1-3\exp(-t)$,
	\begin{align*}
	\|z\|^2 &\leq n + c_1(t + \sqrt{nt}),\\
	\|\Pi_{\mathcal{L}} z\|^2 &\geq n - c_2(k + t + \sqrt{nt}),
	\end{align*}
	where $\Pi_{\mathcal{L}}$ is the orthogonal projection on $\mathcal{L}$.
\end{corollary}
\begin{proof}
First, note that $z$ can have non-zero mean. Therefore, we consider centralizing it as $\tilde{z} = z - \mathbb{E}[z]$. We can apply Corollary~S.7 in \citet{Bartlett2020} to $\tilde{z}$. Then, for any $t > 0$, with probability at least $1-2\exp(-t)$,
\begin{align*}
    \|\tilde{z}\|^2 &\leq n + a'(t + \sqrt{nt}),
\end{align*}
where $a'$ is a constant. Therefore,
\begin{align*}
    \|z\|^2 &\leq \|z - \mathbb{E}[z]\|^2 + \| \mathbb{E}[z]\|^2 \leq \| \mathbb{E}[z]\|^2  + n + a(t + \sqrt{nt}) \leq \| n + a(t + \sqrt{nt})..
\end{align*}

Similarly, for any $t> 0$, with probability at least $1-\exp(-t)$,
\begin{align*}
    \|\Pi_{\mathcal{L}} \tilde{z}\|^2 \geq n - a(2k + 4t + c\max(t, \sqrt{nt})).
\end{align*}
We also have
\begin{align*}
    \|\Pi_{\mathcal{L}} \mathbb{E}[z]\|^2 + \|\Pi_{\mathcal{L}} z\|^2 \geq \|\Pi_{\mathcal{L}} \tilde{z}\|^2.
\end{align*}
Therefore, with probability at least $1-3\exp(-t)$
	\begin{align*}
	\|z\|^2 \leq& n + c_1\max(t, \sqrt{nt}),\\
	\|\Pi_{\mathcal{L}} z\|^2 \geq& \|\tilde{z}\|^2 -  \sigma^2(2k+4t) - \|\Pi_{\mathcal{L}} \mathbb{E}[z]\|^2\\
	\geq& n -  c_2(2k + 4t + c\max(t, \sqrt{nt})).
	\end{align*}
\end{proof}

\section{Proof of Lemma \ref{lem:basic}}
\begin{proof}
First, we have
\begin{align*}
    \hat{y} = \frac{\mathbbm{1}[d=1]y_1}{p(d = 1| x)} + \frac{\mathbbm{1}[d=0]y_0}{p(d = 0| x)}.
\end{align*}
From Assumptions~\ref{asmp:coherent},
\begin{align*}
    \mathbb{E}\left[\hat{y}\right] &= \mathbb{E}\left[\frac{\mathbbm{1}[d=1]y_1}{p(d = 1| x)} - \frac{\mathbbm{1}[d=0]y_0}{p(d = 0| x)}\right]=  \mathbb{E}\left[\frac{p(d=1 | x)\mathbb{E}\left[y_1|x\right]}{p(d = 1| x)} - \frac{p(d=0 | x)\mathbb{E}\left[y_0| x\right]}{p(d = 0| x)}\right] = \mathbb{E}[y_1] - \mathbb{E}[y_0].
\end{align*}
From Assumption~\ref{asmp:basic}, $\mathbb{E}[y_1] = \mathbb{E}[y_0] = 0$. Therefore, $\mathbb{E}\left[\hat{y}\right] = 0$. Thus, Statement~1' holds.

The conditional noise variance can be written as 
\begin{align*}
    &\mathbb{E}\left[\left(\hat{y} - x^\top\theta^*\right)^2 | x\right]\\
    &= \mathbb{E}\left[\hat{y}^2| x\right] - \left(x^\top\theta^*\right)^2\\
    &= \mathbb{E}\left[\frac{\mathbbm{1}[d=1]y^2_1}{p^2(d = 1| x)} + 2\frac{\mathbbm{1}[d=1]\mathbbm{1}[d=0]y_1 y_0}{p(d = 1| x)p(d = 0| x)} + \frac{\mathbbm{1}[d=0]y^2_0}{p^2(d = 0| x)}| x\right] - \left(\tau^*(x)\right)^2\\
    &= \frac{p(d = 1| x)\mathbb{E}\left[y^2_1|x\right]}{p^2(d = 1| x)} + \frac{p(d = 0| x)\mathbb{E}\left[y^2_0|x\right]}{p^2(d = 0| x)}- \left(\tau^*(x)\right)^2\\
    &= \frac{\mathbb{E}\left[y^2_1|x\right]}{p(d = 1| x)} + \frac{\mathbb{E}\left[y^2_0|x\right]}{p(d = 0| x)}- \left(x^\top\theta^*_1 - x^\top\theta^*_0\right)^2\\
    &\geq \frac{\mathbb{E}\left[y^2_1|x\right]}{p(d = 1| x)} + \frac{\mathbb{E}\left[y^2_0|x\right]}{p(d = 0| x)}- \left(x^\top\theta^*_1\right)^2 - \left(x^\top\theta^*_0\right)^2\\
    &= \frac{\mathbb{E}\left[y^2_1|x\right]}{p(d = 1| x)} + \frac{\mathbb{E}\left[y^2_0|x\right]}{p(d = 0| x)} -  \mathbb{E}[y^2_1 | x] - \mathbb{E}[y^2_0 | x] + \mathbb{E}[y^2_1 | x] + \mathbb{E}[y^2_0 | x] - \left(x^\top\theta^*_1\right)^2 - \left(x^\top\theta^*_0\right)^2\\
    &= \frac{\mathbb{E}\left[(1 - p(d = 1| x))y^2_1|x\right]}{p(d = 1| x)} + \frac{(1 - p(d = 0| x))\mathbb{E}\left[y^2_0|x\right]}{p(d = 0| x)} + \Expect\left[(y_1-x^\top\theta^*_1)^2\middle|x\right] + \Expect\left[(y_0-x^\top\theta^*_0)^2\middle|x\right]\\
    &\geq 2\sigma^2,
\end{align*}
where we use $\mathbb{E}[\hat{y} | x] = x^\top\theta^* = \tau^*(x)$, $\mathbbm{1}[d=1]\mathbbm{1}[d=0] = 0$, and $\Expect\left[(y_a-x^\top\theta^*_a)^2\middle|x\right] \geq \sigma^2$ (from Assumption~\ref{asmp:basic}).  Thus, Statement~3' holds.

Proposition 2.5.2 in \citet{vershynin_2018} states that for a random variable $R$, $\mathbb{E}[\exp(R^2)] < \infty$ is equivalent to $R$ being a sub-Gaussian random variable. Here, from Assumption~\ref{asmp:coherent} ($1/p(d=a| x) < \infty$), 
\begin{align*}
    &\mathbb{E}\left[\exp\left(\left(\hat{y} - x^\top\theta^*\right)^2\right)\right]\\
    &= \mathbb{E}\left[\exp\left(\left(\frac{\mathbbm{1}[d=1]y_1}{p(d = 1| x)} - \frac{\mathbbm{1}[d=0]y_0}{p(d = 0| x)} - x^\top\theta^*\right)^2\right)|x\right]\\
    &= \mathbb{E}\Bigg[\exp\Bigg(\Bigg(\frac{\mathbbm{1}[d=1]y_1}{p(d = 1| x)} - \frac{\mathbbm{1}[d=0]y_0}{p(d = 0| x)} - \frac{\mathbbm{1}[d=1]x^\top\theta^*_1}{p(d = 1| x)} + \frac{\mathbbm{1}[d=0]x^\top\theta^*_0}{p(d = 0| x)}\\
    &\ \ \ \ \ \ \ \ \ \ \ \ \ \ \ \ \ \ \ \ \ \ \ \ \ \ \ \ \ \ \ \ \ \ \ \ \ \ \ \ \ \ \ \ \ \ \ \ \ \ \ \ \ \ \ \ +  \frac{\mathbbm{1}[d=1]x^\top\theta^*_1}{p(d = 1| x)} - \frac{\mathbbm{1}[d=0]x^\top\theta^*_0}{p(d = 0| x)} - x^\top\theta^*\Bigg)^2\Bigg)|x\Bigg]\\
    &\leq \mathbb{E}\Bigg[\exp\Bigg(\sum_{a}\Bigg(\frac{\mathbbm{1}[d=a](y_a - x^\top \theta^*_a)}{p(d = a| x)}\Bigg)^2 +  \sum_{a}\Bigg(\frac{\mathbbm{1}[d=a]x^\top\theta^*_a}{p(d = a| x)} - x^\top\theta^*_a\Bigg)^2\Bigg)|x\Bigg] \leq \infty,
\end{align*}
where we also assume that $y_a - x^\top\theta^*_a$ is $\sigma_y^2$-sub-Gaussian, conditionally on $x$. Thus, $\hat{y} - x^\top\theta^*$ is also sub-Gaussian, and Statement~4' holds.
\end{proof}

\section{Proof of Lemma \ref{lem:unbiased_yhat}}
\begin{proof}
\begin{align*}
\mathbb{E}\left[ \hat{y}_i | x\right] &= \mathbb{E}\left[\frac{\mathbbm{1}[d_i = 1]y_i}{p(d=1| x)} - \frac{\mathbbm{1}[d_i = 0]y_i}{1 - p(d=1| x)} | x\right] \\
&= \mathbb{E}\left[\frac{\mathbbm{1}[d_i = 1]y_{1,i}}{p(d=1| x)} - \frac{\mathbbm{1}[d_i = 0]y_{0,i}}{1 - p(d=1| x)} | x\right] \\
&=\frac{ \mathbb{E}\left[\mathbbm{1}[d_i = 1] | x\right]\mathbb{E}\left[y_{1,i} | x\right]}{p(d=1| x)} - \frac{\mathbb{E}\left[\mathbbm{1}[d_i = 0] | x\right]\mathbb{E}\left[y_{0,i} | x\right]}{1 - p(d=1| x)}\\
&=\mathbb{E}\left[y_{1,i} | x\right] - \mathbb{E}\left[y_{0,i} | x\right] = \tau^*(x)
\end{align*}
\end{proof}

\section{Proof of Theorem~\ref{thm:main}}
\label{sec:proof}
This section provides the proof of Theorem~\ref{thm:main}. In Lemma~\ref{lemma:bv}, we first decompose the excess risk. Then, the following lemmas (Lemma~\ref{lemma:bias_1_0}--\ref{lem:f}) show the upper or lower bounds of each decomposed term. The proof of each lemma is provided in Sections~\ref{sec:upp_low_trC}--\ref{sec:d_e}. Combining these bounds, Section~\ref{sec:final} completes the proof.

\subsection{Basic Decomposition of the Excess Risk and Associated Lemmas}
First, we decompose the upper and lower bounds of the excess risk into several terms that can be bounded.

\begin{lemma}[Full version of Lemma \ref{lem:bv_simple}]
\label{lemma:bv}
\begin{align}
&R\big(\hat{\theta}^{\mathrm{T\mathchar`-learner}}\big)\nonumber\\
&=\mathbb{E}_{x}\Big[\Big(x^\top \big(\theta^*_1 - \hat{\theta}_1\big)\Big)^2\Big] + \mathbb{E}_{x}\Big[x^\top\big( \theta^*_0 - \hat{\theta}_0 \big)\Big)^2\Big]- 2\mathbb{E}_{x}\Big[\Big( \big(\theta^*_1 - \hat{\theta}_1\big)^\top x_i\Big)\Big(\big( x^\top_i\theta^*_0 - \hat{\theta}_0 \big) \Big)\Big]\nonumber\\
\label{eq:basic1}
&\leq 2\theta^{*\top}_1 B_1 \theta^*_1 + 2\bm{\varepsilon}^\top_1 C_1 \bm{\varepsilon}_1 + 2\theta^{*\top}_0 B_0 \theta^*_0 + 2\bm{\varepsilon}^\top_0 C_0 \bm{\varepsilon}_0- 2\theta^{*\top}_1B_{1,0}\theta^*_0 + 2\theta^{*\top}_1D \bm \varepsilon_0 + 2\bm \varepsilon^\top_1 E\theta^*_0 - 2\bm \varepsilon^\top_1F\bm \varepsilon_0.
\end{align}
and
\begin{align}
\label{eq:basic2}
&\mathbb{E}_{x,\bm{\varepsilon}}\Big[R\big(\hat{\theta}^{\mathrm{T\mathchar`-learner}}\big) \Big] \geq \theta^{*\top}_1 B_1 \theta^*_1 + \sigma^2_1 \mathrm{tr}(C_1) + \theta^{*\top}_0 B_0 \theta^*_0 + \sigma^2_1 \mathrm{tr}(C_0) - 2\theta^{*\top}_1B_{1,0}\theta^*_0,
\end{align}
where 
\begin{align*}
B_1 &= \big(I - X^\top_1(X_1X^\top_1)^{-1}X_1 \big) \Sigma \big(I - X^\top_1(X_1X^\top_1)^{-1}X_1 \big)\\
B_0 &= \big(I - X^\top_0(X_0X^\top_0)^{-1}X_0 \big) \Sigma \big(I - X^\top_0(X_0X^\top_0)^{-1}X_0 \big)\\
B_{1, 0} &= \big(I - X^\top_1(X_1X^\top_1)^{-1}X_1 \big) \Sigma \big(I - X^\top_0(X_0X^\top_0)^{-1}X_0 \big)\\
C_1 &= \big(X_1X^\top_1 \big)^{-1} X_1 \Sigma X^\top_1  \big(X_1X^\top_1 \big)^{-1}\\
C_0 &= \big(X_0X^\top_0 \big)^{-1} X_0 \Sigma X^\top_0  \big(X_0X^\top_0 \big)^{-1}\\
D &= \big(I - X^\top_1(X_1X^\top_1)^{-1}X_1 \big) \Sigma X^\top_0(X_0X^\top_0)^{-1} \\
E &= (X_1X^\top_1)^{-1}X_1 \Sigma \big(I - X^\top_0(X_0X^\top_0)^{-1}X_0 \big) \\
F &= (X_1X^\top_1)^{-1}X_1 \Sigma X^\top_0(X_0X^\top_0)^{-1} 
\end{align*}
\end{lemma}
The proof is shown in Appendix~\ref{appdx:lemma:bv}. 

In the upper bound \eqref{eq:basic1}, $2\theta^{*\top}_1 B_1 \theta^*_1$ and $2\theta^{*\top}_0 B_0 \theta^*_0$ correspond to the biases in predicting $y_1$ and $y_0$, respectively; $2\bm{\varepsilon}^\top_1 C_1 \bm{\varepsilon}_1$ and $2\bm{\varepsilon}^\top_0 C_0 \bm{\varepsilon}_0$ correspond to the variances in predicting $y_1$ and $y_0$, respectively; $- 2\theta^{*\top}_1B_{1,0}\theta^*_0 + 2\theta^{*\top}_1D \bm \varepsilon_0 + 2\bm \varepsilon^\top_1 E\theta^*_0 - 2\bm \varepsilon^\top_1F\bm \varepsilon_0$ appear from the cross term of $\hat{\theta}_1$ and $\hat{\theta}_0$. In the lower bound \eqref{eq:basic2}, some terms used in the upper bound vanish owing to the independentness of the error terms among the treatments.

Next, we consider bounding each term in \eqref{eq:basic1} and \eqref{eq:basic2}. Here, we use the properties that
\[X_0X^\top_1 = \bm{0},\]
and for a matrix $M\in\mathbb{R}^{n\times n}$,
\[X_0MX^\top_1 = \bm{0},\]
where $\bm{0}$ is a $n\times n$ zero matrix.

\begin{lemma}[Upper bounds regarding the terms including \texorpdfstring{$B_{1}$}{TEXT} and \texorpdfstring{$B_{0}$}{TEXT}]\label{lemma:bias_1_0}
For each $a\in\{0, 1\}$, there exists a constant $c$ that depends only on $\sigma_x$, such that for any $1<t<n$, with
probability at least $1-\exp(-t)$,
  \[
    {\theta^*_a}^\top B_a \theta^*_a
        \leq c\|\theta^*_a\|^2\|\Sigma_a\|\max
        \left\{\sqrt{\frac{r_0(\Sigma_a)}{n}},
        \frac{r_0(\Sigma_a)}{n}, \sqrt{\frac{t}{n}} \right\}+ \left\|\Sigma - \zeta^*_a\Sigma_a\right\|\|\theta^*_a\|^2,
  \]
  where $\zeta^*_a = \argmin{\zeta\in\mathbb{R}^+} \left\|\Sigma - \zeta\Sigma_a\right\|$. 
\end{lemma}

The proof is shown in Appendix~\ref{a:upperB}.

\begin{lemma}[Upper and Lower bounds regarding the term including \texorpdfstring{$B_{1,0}$}{TEXT}]\label{lemma:bias_10}
There is a constant $c$ that depends only on $\sigma_x$ such that for any $1<t<n$, with
probability at least $1-\exp(-t)$,
\begin{align*}
\Big|2\theta^{*\top}_1B_{1,0}\theta^*_0\Big| \leq c\left\|\theta^*_1\right\|\left\|\theta^*_0\right\|\left\|\Sigma\right\|\max\left\{\sqrt{\frac{r_0(\Sigma)}{n}}, \frac{r_0(\Sigma)}{n}, \sqrt{\frac{t}{n}} \right\}
\end{align*}
\end{lemma}
The proof is shown in Appendix~\ref{appdx:lemma:bias_10}.

\begin{lemma}[Upper bounds regarding the terms including \texorpdfstring{$C_{1}$}{TEXT} and \texorpdfstring{$C_{0}$}{TEXT}]
\label{lem:bound_c1_c0}
Suppose that, given $X_1$ and $X_0$, $C_1$ and $C_0$ are a.s.~positive semidefinite. Then a.s.~on $X_1$ and $X_0$, with conditional probability at
  least $1-2\exp(-t)$,
\begin{align*}
\bm{\varepsilon}^\top_1 C_1 \bm{\varepsilon}_1\leq c_0\sigma^2\mathrm{tr}(C_1)\qquad \bm{\varepsilon}^\top_0 C_0 \bm{\varepsilon}_0\leq c_1\sigma^2\mathrm{tr}(C_0)
\end{align*}
\end{lemma}
\begin{proof}
From Proposition~\ref{lemma::sub-Gaussian-quadratic-form} (Lemma~S.2 in \citet{Bartlett2020}), since $\|C_a\|\leq \tr(C_a)$ and $\tr(C^2_a) \leq \tr(C_a)^2$ for $a\in\{1,0\}$,
with probability at least $1-\exp(-t)$,
  \begin{align*}
    \bs{\varepsilon_a}^\top C_a\bs{\varepsilon_a}
      & \leq \sigma^2\tr(C_a)(2t+1) + 2\sigma^2\sqrt{\tr(C_a)^2(t^2 + t)}
      \leq (4t+2)\sigma^2\tr(C_a).
  \end{align*}
\end{proof}

\begin{lemma}[Upper bounds regarding the terms including \texorpdfstring{$D$}{TEXT} and \texorpdfstring{$E$}{TEXT}]
\label{lem:bound_d_e}
Almost surely on $X_1$ and $X_0$, with conditional probability at least $1-\exp(-t)$,
\begin{align*}
   \theta^{*\top}_1\Sigma X^\top_0(X_0X^\top_0)^{-1}\bm{\varepsilon}_0 <  \sqrt{2t\sigma^2\left\|\theta^{*\top}_1\Sigma X^\top_0(X_0X^\top_0)^{-1}\right\|^2_2}.
\end{align*}
Similarly, a.s. on $X_1$ and $X_0$, with conditional probability at least $1-\exp(-t)$,
\begin{align*}
   \bm{\varepsilon}_1(X_1X^\top_1)^{-1}X_1 \Sigma\theta^{*}_0 <  \sqrt{2t\sigma^2\left\|\theta^{*\top}_0\Sigma X^\top_1(X_1X^\top_1)^{-1}\right\|^2_2}.
\end{align*}
\end{lemma}
The proof is shown in Appendix~\ref{appdx:lem:bound_d_e}.

\begin{lemma}[Equality regarding the term including \texorpdfstring{$F$}{TEXT}]
\label{lem:f}
\begin{align*}
    \bm \varepsilon^\top_1F\bm \varepsilon_0  = 0
\end{align*}
\end{lemma}

Lemma~\ref{lem:f} holds from $X_1 \Sigma X^\top_0 = 0$.

\subsection{Concentration Inequalities for the Upper and Lower Bounds of \texorpdfstring{$\mathrm{tr}(C_a)$}{TEXT}}
\label{sec:upp_low_trC}

To show the upper and lower bounds of $\mathrm{tr}(C_a)$ for each $a\in\{1,0\}$ in Sections~\ref{sec:upper_trace} and \ref{sec:lower_trace}, we present the associated lemmas in this section.

\begin{lemma}[From Lemma~3 in \citet{Bartlett2020}]\label{lemma::representation_through_z}
  Consider a covariance operator $\Sigma$ with
  $\lambda_{i}=\mu_i(\Sigma)$ and $\lambda_{n}>0$. Write its spectral
  decomposition $\Sigma=\sum_{j} \lambda_{j} v_{j}v_{j}^\top$,
  where the orthonormal $v_{a,j}\in\bbH$ are the eigenvectors
  corresponding to the $\lambda_{j}$. 
  For $i$ with $\lambda_{i} > 0$, define
  $z_{i}=Xv_{i}/\sqrt{\lambda_{i}}$ and $z_{a,i}=X_av_{i}/\sqrt{\lambda_{i}}$. Then 
    \begin{align*}
      \tr\left(C_a\right)
        &= \sum_{i}\left[ \lambda_{i}^2 z_{a,i}^\top
          \left(\sum_{j} \lambda_{j} z_{a,j} z_{a,j}^\top \right)^{-2}z_{a,i}\right],
    \end{align*}
  and these $z_{a,i}\in\Re^n$ satisfies that for all $\lambda \in \bbH$, 
   	  \begin{align*}
    	&\E[\exp(\lambda^\top z_{a,i})] \leq c p(d=1)\exp(\sigma_x^2\|\lambda\|^2/2) + p(d=0),
   	  \end{align*}
  where $\frac{p(d=1|z)}{p(d=1)} \leq c$. Furthermore, for any $i$ with $\lambda_{i} > 0$, we have
  \begin{align*}
    \lambda_{i}^2 z_{a,i}^\top \left(\sum_{j} \lambda_{j} z_{a,j} z_{a,j}^\top
        \right)^{-2}z_{a,i}
      = \frac{\lambda_{i}^2 z_{a,i}^\top  G_{a,-i}^{-2}z_{a,i}}
        {(1 + \lambda_{i} z_{a,i}^\top G_{a,-i}^{-1}z_{a,i})^2},
  \end{align*}
  where $G_{a,-i} = \sum_{j \neq i} \lambda_{j} z_{a,j} z_{a,j}^\top$.
\end{lemma}

The weighted sum of outer products of these sub-Gaussian vectors plays a
central role in the rest of the proof. Define
  \begin{align*}
    G_a & = \sum_{i} \lambda_{i} z_{a,i} z_{a,i}^\top, &
    G_{a,-i} & = \sum_{j \neq i} \lambda_{j} z_{a,j} z_{a,j}^\top, &
    G_{a,k} &= \sum_{i>k}\lambda_{i} z_{a,i} z_{a,i}^\top,
  \end{align*}
where recall that $z_{a,i}\in\Re^n$ are defined in Lemma~\ref{lemma::representation_through_z}. Note that the vector $z_{a,i}$ is independent of the matrix $G_{a,-i}$; therefore,
in the last part of Lemma~\ref{lemma::representation_through_z},
all the random quadratic forms are independent of the points where
those forms are evaluated.

The next step is to replace Lemma~4 in \citet{Bartlett2020} by showing that eigenvalues of $G_a$, $G_{a,-i}$ and $G_{a,k}$ are concentrated. 

\begin{lemma}
\label{lem:lem4_bartlett}
There is a universal constant $c$ such that with probability at least $1-2\exp(-n/c)$,
\begin{align*}
\frac{1}{c}\sum_i \lambda_i - c\lambda_1n\leq \mu_n(G_a)\leq \mu_1(G_a) \leq c\left(\sum_i\lambda_i + \lambda_1 n\right).
\end{align*}
\end{lemma}
\begin{proof}
First, we develop a probabilistic bound on $| v^\top G_{a}v - \sum_i \lambda_i |$ for any $v \in \mathbb{R}^n$,
by applying the Bernstein inequality for a weighted sum of centered sub-exponential random variables (Proposition \ref{cor:Bernstein}).
To this end, we confirm that $v^\top G_{a}v$ is a sub-exponential random variable.

We fix $v \in \mathbb{R}^n$ and rewrite $v^\top G_{a}v$ as
\begin{align*}
    v^\top G_{a}v = \sum_i\lambda_i\left(v^\top z_{a,i} \right)^2,
\end{align*}
then study its centered element $\big(v^\top z_{a,i}\big)^2 - \mathbb{E}[\big(v^\top z_{a,i}\big)^2]$.
We note that for a random variable $R$, $\mathbb{E}[\exp(R^2)] < \infty$ implies that $R$ is a sub-Gaussian random variable (Proposition 2.5.2 in \citet{vershynin_2018}).
For any $i$, we have
\begin{align*}
\mathbb{E}\left[\exp\left( \left(\sqrt{\lambda_i}v^\top z_{a,i}\right)^2\right)\right]
\leq \mathbb{E}\left[\exp\left( \left(\sqrt{\lambda_i}v^\top\mathbbm{1}[d=a] x\right)^2\right)\right] \leq \mathbb{E}\left[\exp\left( \left(\sqrt{\lambda_i}v^\top x \right)^2\right)\right] < \infty.
\end{align*}
The last inequality follows the assumption on the sub-Gaussian property of $x$.
Then, we use Proposition 2.5.2 in \citet{vershynin_2018} again and find that $v^\top z_{a,i}$ is a sub-Gaussian random variable. 
Furthermore, because a random variable $R$ is sub-Gaussian if and only if $R^2$ is sub-exponential (Lemma 2.7.6 in \citet{vershynin_2018}), we find that $(v^\top z_{a,i})^2$ is a sub-exponential random variable.
Because a centered sub-exponential random variable is also sub-exponential (Exercise 2.7.10 in \cite{vershynin_2018}), $(v^\top z_{a,i})^2 - \mathbb{E}[(v^\top z_{a,i})^2]$ is a sub-exponential random variable.

We study the centered version of $\alpha v^\top G_{a}v$ for a fixed $\alpha \in \mathbb{R}$ using the above result.
The sub-exponential random variable $\alpha v^\top G_{a}v$ has the mean 
\begin{align*}
    \alpha\mathbb{E}\left[\sum_i\lambda_i\left(v^\top \mathbbm{1}[d_i = a]z_i \right)^2\right] = \alpha\sum_i \lambda_i \mathbb{E}\left[p(a| z_i) v^\top z_iz^\top_i v\right],
\end{align*}
and from Lemma 2.7.10 in \citet{vershynin_2018}, $\alpha v^\top G_{a}v - \alpha\mathbb{E}[\sum_i\lambda_i(v^\top \mathbbm{1}[d_i = a]z_i )^2]$ is a centered sub-exponential random variable, which is equal to
\begin{align*}
    &\alpha v^\top G_{a}v - \alpha\mathbb{E}\left[\sum_i\lambda_i\left(v^\top \mathbbm{1}[d_i = a]z_i \right)^2\right] \\
    &= \alpha v^\top G_{a}v - \alpha\sum_i \lambda_i \mathbb{E}\left[p(a| z_i) v^\top z_iz^\top_i v\right] \\
    &= \sum_i\lambda_i\left(\left(v^\top \mathbbm{1}[d_i = a]z_i \right)^2 - \alpha\mathbb{E}\left[p(a| z_i) v^\top z_iz^\top_i v\right]\right).\nonumber
\end{align*}
Let $\alpha = \frac{\sum_{i} \lambda_i}{ \sum_i \lambda_i \mathbb{E}\left[p(a| z_i) v^\top z_iz^\top_i v\right] }$, which is positive and finite because the eigenvalue $\lambda_i$ and $\mathbb{E}\left[p(a| z_i) v^\top z_iz^\top_i v\right]$ are non-negative and $p(a| z_i) > 0$.

Then, because $\alpha v^\top G_\alpha v - \sum_i \lambda_i$ is a centered sub-exponential random variable, Proposition~\ref{cor:Bernstein} \citep[Corollary~S.6,][]{Bartlett2020} yields that for some constant $c_2$ with probability at least $1-2\exp(-t)$,
\begin{align}
    \left|\alpha v^\top G_{a}v - \sum_i \lambda_i \right| \leq c_2\max\left\{\lambda_1 t, \sqrt{t\sum\lambda^2_i}\right\}.
\end{align}

Then, for a fixed vector $\tilde{v} = \sqrt{\alpha}v$,
\begin{align*}
    \left| \tilde{v}^\top G_{a}\tilde{v} - \sum_i \lambda_i \right| \leq c_2\max\left\{\lambda_1 t, \sqrt{t\sum\lambda^2_i}\right\}.
\end{align*}
We also denote $\tilde{v}$ by $v$ for brevity.

Second, we improve the above inequality using a uniform bound technique by following the proof of Theorem~4.4.5 in \citet{vershynin_2018}.
Let $\mathcal{N}$ be a $\frac{1}{4}$-net on the sphere $\mathcal{S}^{n-1}$ with respect to the Euclidean distance such that $|\mathcal{N}|\leq 9^n$. We can find such $\frac{1}{4}$-net from Corollary~4.2.13.
Then, from the union bound over the elements of $\mathcal{N}$, for every $v\in\mathcal{N}$,
\begin{align*}
    &\mathbb{P}\left(\left|v^\top G_{a}v - \sum_i \lambda_i \right| \geq c_1 \max\left\{\lambda_1 t, \sqrt{t\sum\lambda^2_i}\right\}\right)\\
    & \leq \sum_{v\in\mathcal{N}}\mathbb{P}\left(\left|v^\top G_{a}v - \sum_i \lambda_i \right| \geq c_1\max\left\{\lambda_1 t, \sqrt{t\sum\lambda^2_i}\right\}\right)\leq 9^n\cdot 2\exp(-t).
\end{align*}

Therefore, we see that with probability $1-2\exp(-t)$, every $v\in\mathcal{N}$ satisfies
\begin{align}
\label{eq:ineq_diff}
    \left|v^\top G_{a}v - \sum_i \lambda_i \right| \leq c_1\max\left\{\lambda_1 (t+n\log(9)), \sqrt{(t+n\log(9))\sum\lambda^2_i}\right\}.
\end{align}

From Proposition~\ref{lemma::norm_by_net}  \citet[Lemma~S.8,][]{Bartlett2020}, with probability at least $1-2\exp(-t)$,
\begin{align*}
    \left\|G_{a} - I_n\sum_i \lambda_i\right\| \leq c_2\left(\lambda_1 (t+n\log(9)) + \sqrt{(t+n\log(9))\sum\lambda^2_i}\right).
\end{align*}

When $t < n/c_3$, we can write $t + n\log(9) \leq c_4 n$, and we have
\begin{align}
    &\lambda_1 (t+n\log(9)) +  \sqrt{(t+n\log(9))\sum_i\lambda^2_i}\nonumber\\
    \label{eq:ineq_eigen2}
    &\leq c_4 \left(\lambda_1 n + \sqrt{n\sum\lambda^2_i}\right)\\
    &\leq c_4 \lambda_1 n + \sqrt{\left(c^2_4\lambda_1 n\right)\sum_i\lambda_i}\nonumber\\
    &\leq c_4 \lambda_1 n + \frac{1}{2}\left(c^2_4\lambda_1 n\right) + \frac{1}{2}\sum_i\lambda_i\nonumber\\
    \label{eq:ineq_eigen}
    &\leq c_5 \lambda_1 n + \frac{1}{2c_2}\sum_i\lambda_i.
\end{align}
Here, we use the inequality of arithmetic and geometric means. 

Finally, we derive the desired upper bound on $\mu_1(G_{a})$ and lower bound on $\mu_n(G_{a})$.
By the definition of a spectral norm, we have
\begin{align*}
     \left\|G_{a} - I_n\sum_i \lambda_i\right\| = \mu_1\left(G_{a} - I_n\sum_i \lambda_i\right).
\end{align*}
For some constant $c_6 > 0$, from \eqref{eq:ineq_eigen}, $\lambda_1 (t+n\log(9)) +  \sqrt{(t+n\log(9))\sum_i\lambda^2_i}\leq c_5 \lambda_1 n + \frac{1}{2c_2}\sum_i\lambda_i$ implies the upper bound
\begin{align*}
\mu_1(G_a) \leq c_6\left(\sum_i\lambda_i + \lambda_1 n\right).
\end{align*}
Let $v$ be an eigenvector corresponding to $\mu_n(G_a)$. Then, by definition and from \eqref{eq:ineq_diff} and \eqref{eq:ineq_eigen2}, with probability at least $1-2\exp(-t)$,
\begin{align*}
     \left|v^\top G_{a}v - \sum_i \lambda_i \right| &=  \left|\mu_n(G_a) - \sum_i \lambda_i \right|\\
     &\leq c_1\max\left\{\lambda_1 (t+n\log(9)), \sqrt{(t+n\log(9))\sum\lambda^2_i}\right\}\\
     &\leq c_4\lambda_1 n + c_4\sqrt{n\sum_i\lambda^2_i}\\
     &\leq c_4\lambda_1 n + c_4\sqrt{n\lambda_1\sum_i\lambda_i}.
\end{align*}

Using the inequality of arithmetic and geometric means, this implies the following lower bound
\begin{align*}
     \mu_n(G_a) \geq \sum_i \lambda_i - c_4\lambda_1 n - \frac{c_4}{2}t\lambda_1 - \frac{c_4}{2}\sum_i\lambda_i.
\end{align*}
Therefore, by appropriately choosing the constant, we conclude the proof.
\end{proof}
By applying a similar step in the proof of Lemma~\ref{lem:lem4_bartlett} for $i > k$ and any $k\geq0$, we also obtain the following corollary.
\begin{corollary}
\label{cor:lem4_bartlett}
There is a universal constant $c$ such that for any $k\geq 0$ with probability at least $1-2\exp(-n/c)$,
\begin{align*}
\frac{1}{c}\sum_{i>k} \lambda_{i} - c\lambda_{k+1}n\leq \mu_n(G_{a,k})\leq \mu_1(G_{a,k}) \leq c\left(\sum_{i>k}\lambda_{i} + \lambda_{k+1} n\right).
\end{align*}
\end{corollary}

Then, by using Lemma~\ref{cor:lem4_bartlett} instead of Lemma~4 in \citet{Bartlett2020}, we obtain the following lemma corresponding to Lemma~5 in \citet{Bartlett2020}.
\begin{lemma}[From Lemma~5 in \citet{Bartlett2020}]\label{lemma::eigvals_of_truncated}
  There are constants $b,c\ge 1$ such that for any
$k\ge 0$, 
with probability at least $1 - 2e^{-n/c}$,
    \begin{enumerate}
      \item for all $i\ge 1$,
        \begin{align*}
       \mu_{k+1}(G_{a,-i}) \le \mu_{{k+1}}(G_a) \le \mu_1(G_{a,k})
           \le c\left(\sum_{j>k}\lambda_j + \lambda_{k+1}n\right),
        \end{align*}
      \item for all $1\le i\le k$,
        \[
          \mu_n(G_a) \ge \mu_n(G_{a,-i}) \ge \mu_n\left(G_{a,k}\right)
          \ge \frac{1}{c}\sum_{j>k} \lambda_j - c\lambda_{k+1} n,
        \]
      \item if $r_k(\Sigma)\ge bn$, then
        \begin{align*}
          \frac{1}{c}\lambda_{k+1} r_k(\Sigma)
            \le \mu_n\left(G_{a,k}\right)
            \le \mu_1(G_{a,k}) \le c\lambda_{k+1} r_k(\Sigma).
    \end{align*}
    \end{enumerate}
\end{lemma}

\begin{proof}
We provide some inequalities as preparation.
From Lemma~\ref{cor:lem4_bartlett}, with probability at least $1-2\exp(-n/c_1)$, the following inequalities hold:
\begin{align}
\frac{1}{c_1}\sum_{j>k} \lambda_j - c_1\lambda_{k+1}n\leq \mu_n(G_{a,k})\leq \mu_1(G_{a,k}) \leq c_1\left(\sum_{j > k}\lambda_j + \lambda_{k+1} n\right). \label{ineq:eigen_k}
\end{align}

Here, the matrix 
\begin{align*}
G_a - G_{a,k} = \sum_{i} \lambda_{i} z_{a,i} z_{a,i}^\top -  \sum_{i>k}\lambda_{i} z_{a,i} z_{a,i}^\top = \sum_{i < k} \lambda_{i} z_{a,i} z_{a,i}^\top 
\end{align*}
ranks at most $k$ because it is the sum of $k$ matrices of rank $1$. Thus, there is a linear space $\mathcal{L}$ of dimension $n-k$ such that, for all $v\in\mathcal{L}$, $v^\top G_a v = v^\top G_{a,k}v \leq \mu_1(G_{a,k})\|v\|^2$ and, therefore, we have
\begin{align}
    \mu_{k+1}(G_a) \leq \mu_1(G_{a,k}). \label{ineq:eigen}
\end{align}

We show the first statement.
Because $G_{a} \gtrsim G_{a,-i}$ holds for any $i$, Lemma~S.11 in \citet{Bartlett2020} gives $\mu_j(G_{a,-i})\leq \mu_j(G_a)$ for all $i$ and $j$ 
Hence, by combining this inequality with \eqref{ineq:eigen_k} and \eqref{ineq:eigen}, the first statement holds.

For the second statement, for $i \leq k$, $G_{a,k} \preceq G_{a, -i}$, all the eigenvalues of $ G_{a, -i}$ are lower bounded by $\mu_n(G_{a,k})$.
Thus, with \eqref{ineq:eigen_k} and \eqref{ineq:eigen}, the second statement holds.

Finally, for the third statement, if $r_k(\Sigma)\ge b n$,
\begin{align*}
  \sum_{j>k}\lambda_j + \lambda_{k+1}n
    & = \lambda_{k+1} r_k(\Sigma) + \lambda_{k+1}n \le \left(1+\frac{1}{b}\right)\lambda_{k+1} r_k(\Sigma), \\
  \frac{1}{c_1}\sum_{j>k}\lambda_j - c_1\lambda_{k+1}n
    & = \frac{1}{c_1}\lambda_{k+1} r_k(\Sigma) - c_1\lambda_{k+1}n \ge \left(\frac{1}{c_1}-\frac{c_1}{b}\right)
      \lambda_{k+1} r_k(\Sigma).
\end{align*}
Choosing $b > c_1^2$ and $c > \max\{c_1+1/c_1,
(1/c_1-c_1/b)^{-1}\}$ and \eqref{ineq:eigen_k} give the
third claim of the lemma.
\end{proof}

\subsection{Upper Bound on \texorpdfstring{$\mathrm{tr}(C_a)$}{TEXT}}
\label{sec:upper_trace}
Lemma~\ref{lemma::eigvals_of_truncated} gives the following upper bound on $\mathrm{tr}(C_a)$, which corresponds to Lemma~6 in \citet{Bartlett2020}.
\begin{lemma}\label{lemma:traceupper}
  There are constants $b,c\ge 1$ such that 
  if $0\le k\le n/c$, $r_k(\Sigma)\ge bn$, and $m\le k$ 
then  with probability at least 
$1 - 7e^{-n/c}$,
    \begin{align*}
      \tr(C_a)
        &\le c \left(\frac{m}{n} + n \frac{\sum_{i>m}\lambda_i^2}
            {\left(\lambda_{k+1}r_k(\Sigma)\right)^2} \right).
    \end{align*}
\end{lemma}

\begin{proof}
Fix $b$ as in Lemma~\ref{lemma::eigvals_of_truncated}. By Lemma~\ref{lemma::representation_through_z}, 
\begin{align}
\mathrm{tr}(C_a) &= \sum_i \lambda^2_i z^\top_{a,i} G^{-2}_a z_{a,i}\nonumber\\
\label{eq:bound}
&=\sum^m_{i=1}\frac{\lambda_{i}^2 z_{a,i}^\top  G_{a,-i}^{-2}z_{a,i}}{(1 + \lambda_{i} z_{a,i}^\top G_{a,-i}^{-1}z_{a,i})^2} + \sum_{i > m}\lambda^2_i z^\top_i G^{-2}_a z_{a,i}.
\end{align}
First, we consider bounding the first term: the sum up to $m$.
If $r_k(\Sigma)\ge bn$, Lemma~\ref{lemma::eigvals_of_truncated} shows that with probability at least 
$1-2e^{-n/{c_1}}$, for all $i\le k$, we have an upper bound on $\mu_{n}(G_{a,-i})$:
\[\mu_{n}(G_{a,-i})\ge \lambda_{k+1}r_k(\Sigma)/c_1;\]
for all $i$, we have a lower bound on on $\mu_{n}(G_{a,-i})$:
\[\mu_{k+1}(G_{a,-i})\le c_1\lambda_{k+1}r_k(\Sigma).\]
The lower bound on $\mu_{n}(G_{a,-i})$ implies that for all $z \in \R^n$ and $1\le i \leq m$,
\begin{align*}
z^\top G_{a,-i}^{-2} z
  &\leq
  \frac{c_1^2\|z\|^2}{\left(\lambda_{k+1}r_k(\Sigma)\right)^2},
\end{align*}
and the upper bound on $\mu_{k+1}(G_{a,-i})$ gives
\begin{align*}
z^\top G_{a,-i}^{-1} z
  &\geq \left(\Pi_{\mathcal{L}_i} z\right)^\top
    G_{a,-i}^{-1} \Pi_{\mathcal{L}_i} z 
  \geq\frac{\|\Pi_{\mathcal{L}_i} z\|^2}{c_1\lambda_{k+1}r_k(\Sigma)}, &
\end{align*}
where 
$\mathcal{L}_i$ is the span of the
$n-k$ eigenvectors of $G_{a,-i}$ corresponding to the smallest $n-k$
eigenvalues. Recall that $\Pi_{\mathcal{L}}$ is the orthogonal projection on $\mathcal{L}$.
Then, for $i\le m$,
\begin{align}
\label{eq::bound_on_Cterm}
\frac{\lambda_i^2 z_{a,i}^\top  G_{a,-i}^{-2}z_{a,i}}{(1 + \lambda_i z_i^\top G_{a,-i}^{-1}z_{a,i})^2}
  & \leq \frac{z_{a,i}^\top  G_{a,-i}^{-2}z_{a,i}}{(z_{a,i}^\top G_{a,-i}^{-1}z_{a,i})^2}
  \leq c_1^4\frac{\|z_{a,i}\|^2}{\|\Pi_{\mathcal{L}_i}z_{a,i}\|^4}.
\end{align}
Next, we apply Corollary~\ref{cor::cor norm of projection} $m$ times, together with a union bound,
to show that with probability at least $1-3\exp(-t)$, for all $1\le i\le m$ and constants $u_1, u_2$,
\begin{align}
\|z_{a,i}\|^2
  & \leq n + u_1(t + \ln k + \sqrt{n(t + \ln k)})\le c_2n, \label{eqn:zinequality1} \\
\|\Pi_{\mathcal{L}_i}z_{a,i}\|^2
  & \geq n-u_2(k + t + \log k + \sqrt{n(t + \ln k)}) \ge n/c_3,        \label{eqn:zinequality2}
\end{align}
provided that  $t < n/c_0$ and $c > c_0$ for some sufficiently
large $c_0$ (note that $c_2$ and $c_3$ only depend on $c_0$, $a$
and $\sigma_x$, and we can still take $c$ to be sufficiently large at the end
without changing $c_2$ and $c_3$).
Combining~\eqref{eq::bound_on_Cterm},~\eqref{eqn:zinequality1},
and~\eqref{eqn:zinequality2}, with probability at least
$1-5e^{-n/c_0}$,
\begin{align*}
\frac{\lambda_i^2 z_{a,i}^\top
    G_{a,-i}^{-2}z_{a,i}}{(1 + \lambda_i z_{a,i}^\top G_{a,-i}^{-1}z_{a,i})^2}
  \leq \frac{c_4}{n}.
\end{align*}
Then, we have
\begin{align*}
\sum_{i=1}^m\frac{\lambda_i^2 z_{a,i}^\top
    G_{a,-i}^{-2}z_{a,i}}{(1 + \lambda_i z_{a,i}^\top G_{a,-i}^{-1}z_{a,i})^2}
  \leq c_4\frac{m}{n}.
\end{align*}
 
Second, consider bounding the second sum in~\eqref{eq:bound}; that is,
\begin{align*}
    \sum_{i > m}\lambda^2_i z^\top_{a,i} G^{-2}_a z_{a,i}
\end{align*}
Lemma~\ref{lemma::eigvals_of_truncated} shows that, on the same high
probability event that we consider in bounding the first half of the sum,
$\mu_n(G_a)\ge\lambda_{k+1}r_k(\Sigma)/c_1$. Hence,
  \begin{align*}
    \sum_{i>m}\lambda_i^2 z_{a,i}^\top G^{-2}_az_{a,i}
      &\le \frac{c_1^2\sum_{i>m}\lambda_i^2\|z_{a,i}\|^2}
        {\left(\lambda_{k+1}r_k(\Sigma)\right)^2}.
  \end{align*}
Note that $\sum_{i>m}\lambda_i^2\|z_{a,i} - \mathbb{E}[z_{a,i}]\|^2$ is the weighted sum of the centered sub-exponential random
variables, with the weights given by $\lambda_i^2$ in blocks of size $n$. Then, using Proposition~\ref{cor:Bernstein}, we can bound $\sum_{i>m}\lambda_i^2\|z_{a,i} - \mathbb{E}[z_{a,i}]\|^2$. Thus, Proposition~\ref{cor:Bernstein} implies that with probability at least $1-2\exp(-t)$, for some constants $c_3, c_4, c_5$,
\begin{align*}
\sum_{i>m}\lambda_i^2\|z_{a,i}\|^2 & = \sum_{i>m}\lambda_i^2\|z_{a,i}- \mathbb{E}[z_{a,i}]+ \mathbb{E}[z_{a,i}]\|^2\\
& \le 2\sum_{i>m}\lambda_i^2\|z_{a,i}- \mathbb{E}[z_{a,i}]\|^2 + 2\sum_{i>m}\lambda_i^2\|\mathbb{E}[z_{a,i}]\|^2\\
&\le c_3n\sum_{i>m}\lambda_i^2
+ c_4\max\left(\lambda_{m+1}^2t,
\sqrt{tn\sum_{i>m}\lambda_i^4}\right) \\
& \le c_3n\sum_{i>m}\lambda_i^2
+ c_4\max\left(t\sum_{i>m}\lambda_i^2,
\sqrt{tn}\sum_{i>m}\lambda_i^2\right) \\
\label{eq:lemma:res2}
& \le c_5n\sum_{i>m}\lambda_i^2,
\end{align*}
because $t < n/c_0$ and $\mathbb{E}[z_{a,i}]$ is a constant $n$-dimensional vector. Combining the above gives
  \begin{align*}
    \sum_{i>m}\lambda_i^2 z_{a,i}^\top G^{-2}_az_{a,i}
      &\le c_{6}n\frac{\sum_{i>m}\lambda_i^2}{\left(\lambda_{k+1}r_k(\Sigma)\right)^2}.
  \end{align*}
 Finally, putting both parts together and taking $c > \max\{c_0,
 c_4, c_6\}$ gives the lemma.
\end{proof}

\subsection{Lower Bound on \texorpdfstring{$\mathrm{tr}(C_a)$}{TEXT}}
\label{sec:lower_trace}
Next, we derive the lower bound on $\mathrm{tr}(C_a)$. We restate Lemma~8 and 9 in \citet{Bartlett2020} as follows:
\begin{lemma}\label{lemma:singletermlower}
  There is a constant $c$ such that for any  $i\ge 1$ with $\lambda_i > 0$, and any $0\le k\le n/c$, 
 with probability  at least 
$1-5e^{-n/c}$,
    \begin{align*}
      \frac{\lambda_i^2 z_i^\top G_{a,-i}^{-2}z_i}
          {(1 + \lambda_i z_i^\top G_{a,-i}^{-1}z_i)^2}  \ge 
   \frac{1}{cn} \left(1+\frac{\sum_{j >k} \lambda_j
          + n\lambda_{k+1}}{n\lambda_i}\right)^{-2}.
    \end{align*}
\end{lemma}

\begin{lemma}\label{lemma::sum_of_pos}
  Suppose $n\leq \infty$ and $\{\eta_i\}_{i = 1}^n$ is a sequence
  of non-negative random variables, $\{t_i\}_{i=1}^n$ is a sequence
  of non-negative real numbers (at least one of which is strictly
  positive) such that for some $\delta \in (0,1)$ and any
  $i \leq n$, $\Pr(\eta_i > t_i) \geq 1 - \delta$. Then
    \[
      \Pr\left(\sum_{i=1}^n\eta_i \geq \frac12\sum_{i=1}^n t_i\right)
        \geq 1 - 2\delta.
    \]
\end{lemma}
We can show these lemmas as well as those in \citet{Bartlett2020} by replacing Corollary~1 in \citet{Bartlett2020} with Corollary~\ref{cor::cor norm of projection}.

From Lemmas~\ref{lemma::representation_through_z}, \ref{lemma:singletermlower} and~\ref{lemma::sum_of_pos}, we can obtain the following lemma corresponding to Lemma~10 in \citet{Bartlett2020}. 
\begin{lemma}\label{lemma:lowerbound}
  There are constants $c$ such that for any $0\le k\le n/c$ and any $ b > 1$
  with probability at least 
$1-10e^{-n/c}$,
  \begin{enumerate}
    \item If $r_k(\Sigma)<bn$, then $\tr(C_a)\ge\frac{k+1}{c b^2n}$.
    \item If $r_k(\Sigma)\ge bn$, then
        \[
          \tr(C_a)\ge \frac{1}{cb^2} \min_{m\le k}\left(
          \frac{m}{n} + \frac{b^2n \sum_{i>m} \lambda_i^2}
          {\left(\lambda_{k+1}r_k(\Sigma)\right)^2}\right).
        \]
  \end{enumerate}
  In particular, if all choices of $k\le n/c$ give $r_k(\Sigma)<bn$, then
  $r_{n/c}(\Sigma)<bn$ implies that with probability at least 
$1-10e^{-n/c}$,
  $\tr(C_a)= \Omega_{\sigma_x}( 1 )$.
\end{lemma}
We can show the lemma as well as Lemma~10 in \citet{Bartlett2020} by replacing Lemmas~3, 8, and 9 in \citet{Bartlett2020} with Lemmas~\ref{lemma::representation_through_z}, \ref{lemma:singletermlower} and~\ref{lemma::sum_of_pos}.

\subsection{Upper Bounds regarding  \texorpdfstring{$\left\|\theta^{*\top}_1\Sigma X^\top_0(X_0X^\top_0)^{-1}\right\|_2$}{TEXT} and \texorpdfstring{$\left\|\theta^{*\top}_0\Sigma X^\top_1(X_1X^\top_1)^{-1}\right\|_2$}{TEXT}}
\label{sec:d_e}
Next, we show the upper bounds of $\left\|\theta^{*\top}_1\Sigma X^\top_0(X_0X^\top_0)^{-1}\right\|_2$ and $\left\|\theta^{*\top}_0\Sigma X^\top_1(X_1X^\top_1)^{-1}\right\|_2$. The proof uses the result presented in Section~\ref{sec:upper_trace}.
\begin{lemma}\label{lemma:l2upper}
For each $a\in\{1,0\}$, there are constants $b,c\ge 1$ such that 
  if $0\le k\le n/c$, $r_k(\Sigma)\ge bn$, and $m\le k$ 
then  with probability at least 
$1 - 7e^{-n/c}$,
    \begin{align*}
      \left\|\theta^{*\top}_a\Sigma X^\top_{1-a}(X_{1-a}X^\top_{1-a})^{-1}\right\|_2
        &\leq  c\left\|\theta^{*}_a\right\|_2 \sqrt{\left(\frac{m}{n} + n \frac{\sum_{i>m}\lambda_i^2}
            {\left(\lambda_{k+1}r_k(\Sigma)\right)^2} \right)}.
    \end{align*}
\end{lemma}

\begin{proof}
We consider a case with $a=1$. We can also show another case with $a=0$. 

By applying a step similar to that used in Lemma~\ref{lemma::representation_through_z}, we decompose the target value as
\begin{align*}
&\left\|\theta^{*\top}_1\Sigma X^\top_0(X_0X^\top_0)^{-1}\right\|_2\\
&\le \left\|\theta^{*\top}_1\sum_i \lambda_i\sqrt{\lambda_i}v_iz^\top_{0,i}\left( \sum_j\lambda_jz_{0,j}z^\top_{0,j}\right)^{-1}\right\|_2\\
&= \left\|\theta^{*\top}_1\sum_i \lambda_i\sqrt{\lambda_i}v_iz^\top_{0,i}\left( \lambda_jz_{0,j}z^\top_{0,j} + G_{0,-j}\right)^{-1}\right\|_2.
\end{align*}
The Sherman--Morrison--Woodbury formula gives
\begin{align*}
&\left( \lambda_iz_{0,i}z^\top_{0,i} + G_{0,-i}\right)^{-1} =  G^{-1}_{0,-i} - G^{-1}_{0,-i}\sqrt{\lambda_i}z_{0,i}\left(1 + \lambda_iz^\top_{0,i}G^{-1}_{0, -i}z_{0,i}\right)^{-1}z^\top_{0,i}\sqrt{\lambda_i}G^{-1}_{0,-i}.
\end{align*}
Therefore, 
\begin{align*}
&\sqrt{\lambda_i}v_iz^\top_{0,i}\left( \lambda_iz_{0,i}z^\top_{0,i} + G_{0,-i}\right)^{-1}\\
&=  \sqrt{\lambda_i}v_iz^\top_{0,i}\left(G^{-1}_{0,-i} - G^{-1}_{0,-i}\sqrt{\lambda_i}z_{0,i}\left(1 + \lambda_iz^\top_{0,i}G^{-1}_{0, -i}z_{0,i}\right)^{-1}z^\top_{0,i}\sqrt{\lambda_i}G^{-1}_{0,-i}\right)\\
&= v_i\left(z^\top_{0,i}\sqrt{\lambda_i}G^{-1}_{0,-i} -  z^\top_{0,i}G^{-1}_{0,-i}\lambda_iz_{0,i}\left(1 + \lambda_iz^\top_{0,i}G^{-1}_{0, -i}z_{0,i}\right)^{-1}z^\top_{0,i}\sqrt{\lambda_i}G^{-1}_{0,-i}\right)\\
&=  v_i\left(1 - z^\top_{0,i}G^{-1}_{0,-i}\lambda_iz_{0,i}\left(1 + \lambda_iz^\top_{0,i}G^{-1}_{0, -i}z_{0,i}\right)^{-1}\right)z^\top_{0,i}\sqrt{\lambda_i}G^{-1}_{0,-i}\\
&=  v_i\left(1 + \lambda_iz^\top_{0,i}G^{-1}_{0, -i}z_{0,i}\right)^{-1}z^\top_{0,i}\sqrt{\lambda_i}G^{-1}_{0,-i}.
\end{align*}
In this study, we use $1 - z^\top_{0,i}G^{-1}_{0,-i}\lambda_iz_{0,i}\left(1 + \lambda_iz^\top_{0,i}G^{-1}_{0, -i}z_{0,i}\right)^{-1} = \left(1 + \lambda_iz^\top_{0,i}G^{-1}_{0, -i}z_{0,i}\right)^{-1}$ because 
\begin{align*}
    &\left(1 + \lambda_iz^\top_{0,i}G^{-1}_{0, -i}z_{0,i}\right)\left(1 - z^\top_{0,i}G^{-1}_{0,-i}\lambda_iz_{0,i}\left(1 + \lambda_iz^\top_{0,i}G^{-1}_{0, -i}z_{0,i}\right)^{-1}\right)\\
    &= \left(1 + \lambda_iz^\top_{0,i}G^{-1}_{0, -i}z_{0,i}\right) - \left(1 + \lambda_iz^\top_{0,i}G^{-1}_{0, -i}z_{0,i}\right)\left(z^\top_{0,i}G^{-1}_{0,-i}\lambda_iz_{0,i}\left(1 + \lambda_iz^\top_{0,i}G^{-1}_{0, -i}z_{0,i}\right)^{-1}\right)\\
    & = \left(1 + \lambda_iz^\top_{0,i}G^{-1}_{0, -i}z_{0,i}\right) - z^\top_{0,i}G^{-1}_{0,-i}\lambda_iz_{0,i} = 1
\end{align*}
Thus, we have
\begin{align*}
&\left\|\theta^{*\top}_1\Sigma X^\top_0(X_0X^\top_0)^{-1}\right\|_2\\
&=\sqrt{\theta^{*\top}_1\Sigma X^\top_0(X_0X^\top_0)^{-2}X_0\Sigma}\\
&\le \sqrt{\|\theta^*_1\|^2_2\left\|\Sigma X^\top_0(X_0X^\top_0)^{-2}X_0\Sigma\right\|}\\
&\le \sqrt{\|\theta^*_1\|^2_2\mathrm{tr}\left(\Sigma X^\top_0(X_0X^\top_0)^{-2}X_0\Sigma\right)}.
\end{align*}
Then, because $X_0X^\top_0 = \sum_i \lambda_i z_{0,i}z^\top_{0,i}$ and $X_0\Sigma = \sum_i\lambda_i\sqrt{\lambda_i}z_{0,i}v^\top_i$ from $\sqrt{\lambda_{i}}z_{0,i}=X_0v_{i}$ and $\Sigma = \sum_i\lambda_iv_i v^\top_i$, 
\begin{align*}
&\mathrm{tr}\left(\Sigma X^\top_0(X_0X^\top_0)^{-2}X_0\Sigma\right)\\
&=\mathrm{tr}\left(\left(\sum_j \lambda_j z_{0,j}z^\top_{0,j}\right)^{-2}\sum_i\lambda^3_iz_{0,i}v^\top_iv_iz^\top_{0,i}\right)\\
&=c\sum_i\mathrm{tr}\left(\left(\sum_j \lambda_i z_{0,j}z^\top_{0,j}\right)^{-2}\lambda^2_iz_{0,i}v^\top_iv_iz^\top_{0,i}\right)\\
&=c\sum_i\mathrm{tr}\left(\lambda^2_iz^\top_{0,i}\left(\sum_j \lambda_i z_{0,j}z^\top_{0,j}\right)^{-2}z_{0,i}\right)\\
&=c\sum_i\lambda^2_iz^\top_{0,i}\left(\sum_j \lambda_i z_{0,j}z^\top_{0,j}\right)^{-2}z_{0,i},
\end{align*}
where we use $v^\top_i v_i = 1$ and $\max lambda_i \leq c$ for some constant $c > 0$. 
Then, by applying the same step in the proof of Lemma~\ref{lemma:traceupper} to $\sum_i\lambda^2_iz^\top_{0,i}\left(\sum_j \lambda_i z_{0,j}z^\top_{0,j}\right)^{-2}z_{0,i}$,
\begin{align*}
&\left\|\theta^{*\top}_1\Sigma X^\top_0(X_0X^\top_0)^{-1}\right\|_2\\
&\le \sqrt{\|\theta^*_1\|^2_2\mathrm{tr}\left(\Sigma X^\top_0(X_0X^\top_0)^{-2}X_0\Sigma\right)}\\
&\le c\|\theta^*_1\|_2\sqrt{\frac{m}{n} + n \frac{\sum_{i>m}\lambda_i^2}
            {\left(\lambda_{k+1}r_k(\Sigma)\right)^2}}.
\end{align*}
This concludes the proof.
\end{proof}

\subsection{Final Step for Proof of the Upper Bound}
\label{sec:final}
To complete the proof of Theorem~\ref{thm:main}, we combine Lemmas~\ref{lemma:bias_1_0}--\ref{lemma:bestk} with Lemma~\ref{lemma:bv}.
We set $b$ in Lemmas~\ref{lemma:lowerbound}--\ref{lemma:l2upper} and Theorem~\ref{thm:main} to the constant $b$ from  Lemma~\ref{lemma:traceupper}. Let $c_1$ be the maximum of the constants $c$ from Lemmas~\ref{lemma:lowerbound} and~\ref{lemma:traceupper}. 

By using Lemma~\ref{lemma:lowerbound}, we consider the lower bound based on the value of $k$. If there is no $k\le n/c$ such that $r_k(\Sigma)\ge bn$, then Lemma~\ref{lemma:lowerbound} implies that $\tr(C_a)\ge\frac{k+1}{c b^2n}$. Then, by combining it with Lemmas~\ref{lemma:bv} and \ref{lem:bound_c1_c0}, we can obtain the lower bound of the expected excess risk as $\Omega(\sigma^2)$,
which proves the first lower bound of Theorem~\ref{thm:main} for large $k^*$: suppose $\delta<1$ with $\log(1/\delta)<n/c$.
 If $k^* \geq n/c_1$, then 
 \[\Expect R(\hat\theta)\ge \sigma^2/c.\] 
If there exist some $k\le n/c$ such that $r_k(\Sigma)\ge bn$, then from Lemmas~\ref{lemma:traceupper} and~\ref{lemma:lowerbound}, the upper and lower
bounds of Lemmas~\ref{lemma:traceupper} and~\ref{lemma:lowerbound} regarding the terms, including $C_1$ and $C_0$, are
constant multiples of
  \[
    \min_{m\le k} \left(
      \frac{m}{n} + n\frac{\sum_{i>m}\lambda_i^2}
      {\left(\lambda_{k+1}r_k(\Sigma)\right)^2}\right);
  \]
from Lemmas~\ref{lem:bound_d_e} and \ref{lemma:l2upper}, the upper and lower
bounds regarding the term, including $D$ and $E$, are also
constant multiples of
\[
    \min_{m\le k} \sqrt{
      \frac{m}{n} + n\frac{\sum_{i>m}\lambda_i^2}
      {\left(\lambda_{k+1}r_k(\Sigma)\right)^2}};
\]

Note that by Lemma~\ref{lemma::eigvals_of_truncated}, for any qualifying value of $k$, the smallest eigenvalue
of $G_a$ is within a constant factor of $\lambda_{k+1}r_k(\Sigma)$. Thus, any two choices of $k$ satisfying $k\le n/c$ and $r_k(\Sigma)\ge
bn$ must have values of $\lambda_{k+1}r_k(\Sigma)$ within constant factors.  The smallest such $k$ simplifies the bound on $\tr(C)$, as the following proposition in \citet{Bartlett2020} shows. 

\begin{proposition}
[Lemma~11 in \citet{Bartlett2020}]
\label{lemma:bestk}
For any $b\geq 1$ and $k^* := \min \{k: r_k(\Sigma)\geq bn\}$, if $k^* <\infty$, we have
\begin{align*}
\min_{m\leq k^*}\left(\frac{m}{bn} + \frac{bn\sum_{i > m} \lambda^2_u}{\big(\lambda_{k^*+1}r_{k^*}(\Sigma)\big)^2}\right) = \frac{k^*}{bn} + \frac{bn\sum_{i > k^*}\lambda^2_i}{\big(\lambda_{k^*+1}r_{k^*}(\Sigma)\big)^2} = \frac{k^*}{bn} + \frac{bn}{R_{k^*}(\Sigma)}.
\end{align*}
\end{proposition}

By Proposition~\ref{lemma:bestk}, the lower bound is 
within a constant factor of $\frac{k^*}{n} + \frac{n}{R_{k^*}(\Sigma)}.$

Taking $c$ sufficiently large, and combining
these results with Lemma~\ref{lemma:bv} and the upper bound on the term ${\theta^\ast}^\top B \theta^\ast$ in Lemma~\ref{lemma:bias_10} completes the proofs of Theorems~\ref{thm:main} and \ref{thm:main2}.

\section{Proof of Lemma \ref{lem:equiv_risks}}

\begin{proof}

We show that $R({\theta}) - \tilde{R}({\theta}) = 0$ as
\begin{align*}
      &R(\hat{\theta}^{\mathrm{IPW\mathchar`-learner}}) - \tilde{R}\big(\hat{\theta}^{\mathrm{IPW\mathchar`-learner}}\big)\\
      &= \mathbb{E}_{x,y}\Big[\big(\tilde{y} - x^\top \theta \big)^2 - \big(\tilde{y} - x^\top \theta^* \big)^2\Big] - \mathbb{E}_{x,y}\Big[\big(\hat{y} - x^\top \theta \big)^2 - \big(\hat{y} - x^\top \theta^* \big)^2\Big]\\
      &= \mathbb{E}_{x,y}\Big[\big(\tilde{y} - \hat{y} + \hat{y} - x^\top \theta \big)^2 - \big(\tilde{y} - \hat{y} + \hat{y} - x^\top \theta^* \big)^2\Big] - \mathbb{E}_{x,y}\Big[\big(\hat{y} - x^\top \theta \big)^2 - \big(\hat{y} - x^\top \theta^* \big)^2\Big]\\
      &= \mathbb{E}_{x,y}\Big[\big(\tilde{y} - \hat{y} \big)^2 + 2\big(\tilde{y} - \hat{y} \big) \big(\hat{y} - x^\top \theta \big) - \big(\tilde{y} - \hat{y} \big)^2 - 2\big(\tilde{y} - \hat{y} \big) \big(\hat{y} - x^\top \theta^* \big)\Big]\\
      &= 2 \mathbb{E}_{x,y}\Big[ \big(\tilde{y} - \hat{y} \big) \big(x^\top \big(\theta^* - \theta\big) \big)\Big]\\
      &= 2 \mathbb{E}_{x}\Big[  \mathbb{E}_{y}\Big[ \tilde{y} - \hat{y} |x\Big]\big(x^\top \big(\theta^* - \theta\big) \big)\Big]\\
      &= 2 \mathbb{E}_{x}\Big[   \big(\tau^*(x) - \tau^*(x) \big) \big(x^\top \big(\theta^* - \theta\big) \big)\Big]\\
      &=0.
\end{align*}
Here, we used 
\begin{align*}
    \mathbb{E}\left[\tilde{y} | x\right] = \mathbb{E}\left[y_1 - y_0| x\right] = \tau^*(x).
\end{align*}
\end{proof}

\section{Proof of Theorem \ref{thm:main2}}
\begin{proof}

Because Statements~1', 3', and 4' in Lemma~\ref{lem:basic} correspond to Assumptions~1, 3, and 4 in Definition~1 in \citet{Bartlett2020}, by combining them with 2 and 5 in Assumptions~\ref{asmp:basic}, we can directly apply Theorem~1 of \citet{Bartlett2020} to obtain the following result.
\begin{corollary}[Excess risk upper bounds in the IPW-learner]
\label{cor:main}
For any $\sigma_x$ there are  $b,c, c_1>1$ for which the following holds. Consider a linear regression problem from
Section~\ref{sec:liner_regression_model} and suppose that Assumption \ref{asmp:unconfounded}, \ref{asmp:coherent}, and \ref{asmp:basic} hold.
Suppose $\delta<1$ with $\log(1/\delta)<n/c$.
If $k^* < n/c_1$, then the excess risk (Definition \ref{def:excess}) of the predictor in \eqref{def:predictor} satisfies
   \begin{align*}
   \tilde{R}\big(\hat{\theta}^{\mathrm{IPW\mathchar`-learner}}\big)
       & \le c\|\theta^*\|^2\mathcal{B}_{n,\delta}(\Sigma)+ c\log(1/\delta)\mathcal{V}_n(\Sigma).
   \end{align*}
 with probability at least $1-\delta$.
\end{corollary}
\begin{proof}
In Definition~1 of \citet{Bartlett2020}, we replace Assumptions~1, 3, and 4 with Statements~1', 3', and 4' in Lemma~\ref{lem:basic}. Then, we can define a linear regression problem under these assumptions. Therefore, we can directly apply Theorem~1 in \citet{Bartlett2020} to obtain the upper bound of the new excess risk $\tilde{R}\big(\theta\big)$.
\end{proof}
By combining this corollary with Lemma \ref{lem:equiv_risks}, we obtain the statement.
\end{proof}

\section{Proof of Theorem \ref{theorem:benign_eigenvalues}}

We firstly develop the following lemma on the eigenvalues of $\Sigma_a$.
\begin{lemma}[Lemma \ref{lem:benign_cov_a}]
If $\Sigma$ is a benign covariance, we have
\begin{align}
\max_{a \in \{0,1\}}\sum_{k=1}^\infty \mu_k (\Sigma_a) = o(n)
\end{align}
\end{lemma}
\begin{proof}
Since 
$\sum_{k=1}^\infty \mu_k (\Sigma) = o(n)$
holds for the benign covariance $\Sigma$, showing the positive definiteness of $\Sigma - \Sigma_a$ is sufficient to achieve the statement.

Take $z \in \mathbb{H}$ arbitrary such that $\|z\| = 1$.
\begin{align*}
    z^\top (\Sigma - \Sigma_a)z &= z^\top \Expect[(1  - \mathbbm{1}[d=a])xx^\top] z \\
    &= z^\top \Expect_x[ \Expect_d[1  - \mathbbm{1}[d=a]|x]xx^\top]z \\
    &=z^\top \Expect_x[ (1  - p(d=a|x))xx^\top]z.
\end{align*}
Since $x x^\top$ is a positive semi-definite operator, we obtain
\begin{align*}
    z^\top (1-p(d=a|x)) x x^\top z \geq z^\top (1-\varphi) x x^\top z.
\end{align*}
Hence, we have
\begin{align*}
    z^\top \Expect_x[ (1  - p(d=a|x))xx^\top]z \geq (1-\varphi) z^\top \Expect[x x^\top] z = (1-\varphi) z^\top \Sigma z > 0.
\end{align*}
The last inequality follows the positive definitenes of $\Sigma$ and Assumption \ref{asmp:coherent}.
\end{proof}




\section{Proof of Theorem \ref{prp:benign_eigenvalues}}
\begin{proof}
We obtain the statement by combining Theorem \ref{thm:main2} and Theorem 6 in \cite{Bartlett2020}.
\end{proof}

\end{document}